\documentclass[11pt]{article}
\usepackage{geometry} 
\usepackage{fullpage}
\usepackage{hyperref}
\usepackage{amsmath}
\usepackage{amsfonts}
\usepackage{amsthm}
\usepackage[utf8]{inputenc}
\usepackage{mdwlist}
\usepackage{pifont}
\usepackage{pgf,pgfarrows,pgfnodes}
\usepackage{tikz}
\usetikzlibrary{shapes.geometric}
\usetikzlibrary{shapes.misc}
\usetikzlibrary{decorations.pathmorphing}
\usetikzlibrary{decorations.pathreplacing}
\usetikzlibrary{positioning,backgrounds,shadows}
\newcommand{\pbsc}[1]{\textsc{#1}}
\newcommand{\IE}{\pbsc{Index Erasure}}
\newcommand{\search}{\pbsc{Search}}
\newcommand{\fsearch}{\pbsc{fold search}}
\newcommand{\GI}{\pbsc{Graph Isomorphism}}
\newcommand{\SE}{\pbsc{Set Equality}}
\newcommand{\collision}{\pbsc{Collision}}
\newcommand{\aGamma}{\tilde{\Gamma}}
\newcommand{\alambda}{\tilde{\lambda}}
\newcommand{\aUpsilon}{\tilde{\Upsilon}}
\newcommand{\aW}{\tilde{W}}
\newcommand{\ADV}{\widetilde{\mathrm{ADV}}}
\newcommand{\MADV}{\mathrm{MADV}}
\newcommand{\NADV}{\mathrm{ADV}^\pm}
\newcommand{\junk}{\mathrm{junk}}
\newcommand{\default}{\mathrm{\bar{0}}}
\newcommand{\error}{\mathrm{error}}
\newcommand{\bad}{\mathrm{bad}}
\newcommand{\good}{\mathrm{good}}

\newcommand{\rep}{\mathcal{U}}
\newcommand{\Fset}{F}
\newcommand{\target}{{\rho^{\odot}}}
\newcommand{\psitarget}{\psi^{\odot}}
\def\I{\mathbb{I}}

\newcommand{\tracenorm}[1]{\left\| #1 \right\|_{\mathrm{tr}}}
\newcommand{\fnorm}[1]{\left\| #1 \right\|_{\mathrm{F}}}
\newcommand{\ket}[1]{| #1 \rangle}
\newcommand{\bra}[1]{\langle #1 |}
\newcommand{\braket}[2]{\langle #1 | #2 \rangle}
\newcommand{\ketbra}[2]{| #1 \rangle\!\langle #2 |}
\newcommand{\proj}[1]{| #1 \rangle\!\langle #1 |}
\newcommand{\norm}[1]{\left\| #1 \right\|}
\newcommand{\abs}[1]{\left| #1 \right|}
\newcommand{\Tr}{{\rm tr}}
\newcommand{\inv}[1]{\frac{1}{ #1 }}
\newcommand{\Span}{\mathrm{Span}}
\renewcommand{\Im}{\mathrm{Im}}

\def\AA{{\cal A}}  \def\CC{{\cal C}}  \def\FF{{\cal F}}     \def\II{{\cal I}}  \def\MM{{\cal M}}  \def\OO{{\cal O}}  \def\PP{{\cal P}}   \def\TT{{\cal T}}  \def\WW{{\cal W}}

\newtheorem{thm}{Theorem}
\newtheorem{definition}[thm]{Definition}
\newtheorem{lem}[thm]{Lemma}
\newtheorem{fact}[thm]{Fact}

\title{Symmetry-assisted adversaries for quantum state generation}
\author{
	Andris Ambainis\thanks{University of Latvia; ambainis@lu.lv} \and
	Loïck Magnin\thanks{Universit\'e Paris Diderot and Universit\'e Libre de Bruxelles; loick.magnin@lri.fr}$\ ^{, \ddag}$ \and
	Martin Roetteler\thanks{NEC Laboratories America; $\{$mroetteler,\;jroland$\}$@nec-labs.com} \and
	Jérémie Roland$^\ddag$
}
\date{\today}

\begin{document}

\maketitle

\begin{abstract}
  We introduce a new quantum adversary method to prove lower bounds on
  the query complexity of the quantum state generation problem. This
  problem encompasses both, the computation of partial or total
  functions and the preparation of target quantum states. There has
  been hope for quite some time that quantum state generation might be
  a route to tackle the {\sc Graph Isomorphism} problem. We show that
  for the related problem of {\sc Index Erasure} our method
  leads to a lower bound of $\Omega(\sqrt N)$ which matches an upper
  bound obtained via reduction to quantum search on $N$ elements. This
  closes an open problem first raised by Shi [FOCS'02].

  Our approach is based on two ideas: (i) on the one hand we
  generalize the known additive and multiplicative adversary methods
  to the case of quantum state generation, (ii) on the other hand we
  show how the symmetries of the underlying problem can be leveraged
  for the design of optimal adversary matrices and dramatically
  simplify the computation of adversary bounds.  Taken together, these
  two ideas give the new result for {\sc Index Erasure} by using the
  representation theory of the symmetric group.  Also, the method can
  lead to lower bounds even for small success probability, contrary to the standard adversary method.
  Furthermore, we answer an open question due to \v{S}palek [CCC'08]
  by showing that the multiplicative version of the
  adversary method is stronger than the additive one for any problem. Finally, we prove that the
  multiplicative bound satisfies a strong direct product theorem,
  extending a result by \v{S}palek to quantum state generation
  problems.

\end{abstract}

\thispagestyle{empty}
\setcounter{page}{0}
\newpage
%\tableofcontents

\section*{Introduction}

The query model provides a way to analyze quantum algorithms
including, but not limited to, those of Shor \cite{Sho97} and Grover
\cite{Gro96} as well as quantum walks, quantum counting, and hidden
subgroup problems. Traditionally, in this model the input is a
black-box function which can be accessed via queries and the output is a {\em classical} value. The measure of
complexity of an algorithm is then defined as the number of queries
made by the algorithm. Studying the quantum query
complexity of functions is quite fruitful since the model is simple
enough that one can show tight bounds for several problems and hence
provides some intuition about the power of quantum computing.

In this paper, we study a generalization of the query model to include
problems in which the input is still a black-box function, however,
the output is no longer a classical value but a target {\em
  quantum} state. An example for the resulting quantum state
generation problem is {\sc Index Erasure}. Here we are
given access to an injective function $f:[N]\rightarrow [M]$ and the task is to
prepare the quantum state $\frac{1}{\sqrt{N}}\sum_{x=1}^N \ket{f(x)}$
using as few queries to $f$ as possible. The name ``index erasure''
stems from the observation that while it is straightforward to prepare
the (at first glance perhaps similar looking) state
$\frac{1}{\sqrt{N}} \sum_{x=1}^N \ket{x}\ket{f(x)}$, it is quite
challenging to forget (``erase'') the contents of the first register
of this state which carries the input (``index'') of the function.

In particular, this approach has been considered in~\cite{AT03} to solve statistical
zero knowledge problems, one ultimate goal being to
tackle \GI{} \cite{KST:93}. The quantum state
generation problem resulting from the well-known reduction of \GI{} to
{\sc Index Erasure} would be to generate the uniform superposition of
all the permutations of a graph $\Gamma$:
\[
\ket{\Gamma} = \frac{1}{\sqrt{n!}} \sum_{\pi\in S_n} \ket{\Gamma^\pi}.
\]
By coherently generating this state for two given graphs, one could
then use the standard SWAP-test to check whether the two states are
equal or orthogonal, and therefore decide whether the graphs are
isomorphic or not. Such a method for solving \GI{} would be
drastically different from more standard approaches based on the
reduction to the hidden subgroup problem, and might therefore provide
a way around serious limitations of the coset state
approach~\cite{HMR+06}. There has been hope for quite some time that
quantum state generation might be a route to tackle the {\sc Graph
  Isomorphism} problem, however one of the main results of this paper
is that any approach that tries to generate $\ket{\Gamma}$ without
exploiting further structure\footnote{Indeed, here we assume that the
  only way to access the graph $\Gamma$ would be by querying an oracle that,
  given a permutation $\pi$, returns the permuted graph $\Gamma^\pi$. Note that we assume that $\Gamma$ is rigid which is no loss of generality.} of
the graph cannot improve on the simple $O(\sqrt{n!})$ upper bound via
search. More generally, we are interested in the query complexity of
the quantum state generation problem, in which the amplitudes of the
target quantum state can depend on the given function in an arbitrary
way. Subroutines for quantum state generation might provide a useful
toolbox to design efficient quantum algorithms for a large class of
problems.

\paragraph{Adversaries.} Lower bounds on the quantum query complexity
have been shown for a wide range of (classical in the above sense)
functions. Roughly speaking, currently there are two main ideas for
proving lower bounds on quantum query complexity: the polynomial
method \cite{BBB+97,BBC+98, Aar02, Shi02, Amb03, KSW07} and the adversary method \cite{Amb00}.
The latter method has seen a sequence of variations, generalizations,
and improvements over the past decade including \cite{HNS08, Amb03,
  BS04, LM08}.

The basic idea behind the adversary method and its variations is to
define a progress function that monotonically changes from an initial
value (before any query) to a final value (depending on the success
probability of the algorithm) with one main property: the value of the
progress function changes only when the oracle is queried. Then, a
lower bound on the quantum query complexity of the problem can be
obtained by bounding the amount of progress done by one query.

Different adversary methods were introduced, but they were later
proved to be all equivalent~\cite{SS05}. They rely on optimizing an
adversary matrix assigning weights to different pairs of inputs to the
problem. While originally these methods only considered positive
weights, it was later shown that negative weights also lead to a lower
bound, which can actually be stronger in some cases~\cite{HLS07}. The
relevance of this new adversary method with negative weights, called
\emph{additive}, was made even clearer when it was very recently shown
to be tight for the quantum query complexity of functions in the
bounded-error model~\cite{Rei09,LMRS10}.
%For non-Boolean functions, however, the situation is not so clear. For some problems, it is known that the adversary method gives weaker bounds than the so-called polynomial method~\cite{BBC+98, Aar02, Amb03, KSW07}, or some other ad-hoc techniques~\cite{Amb05, ASW07}. For this reason, outside of the realm of Boolean functions, the quest for an all-powerful lower bound technique is not over. 

Nevertheless, for some problems other methods (such as the polynomial method or other ad-hoc techniques) might be easier to
implement while also leading to strong bounds.
% Among them, we can cite
% the polynomial method~\cite{BBC+98, Aar02, Shi02, Amb03, KSW07} and
% some other ad-hoc techniques~\cite{Amb99, Amb05, ASW07}.
The additive
adversary method also suffers from one main drawback: it cannot prove
lower bounds for very small success probability. To circumvent it,
\v{S}palek introduced the \emph{multiplicative} adversary
method~\cite{Spa08} that generalizes some previous ad-hoc methods~\cite{Amb05, ASW07}. Being able to deal with exponentially small success
probability also allowed to prove a strong direct product theorem for
any function that admits a multiplicative adversary lower
bound~\cite{Amb05,ASW07,Spa08} (note that a similar result has
recently been proved for the polynomial method~\cite{She10}). Roughly
speaking, it means that if we try to compute $k$ independent instances
of a function using less than $O(k)$ times the number of queries
required to compute one instance, then the overall success probability
is exponentially small in $k$. However, \v{S}palek left unanswered the
question of how multiplicative and additive methods relate in the case
of high success probability. In particular, it is unknown whether the
strong direct product theorem extends to the additive adversary
method, and therefore to the quantum query complexity of any function
since this method is known to be tight in the bounded error
model~\cite{LMRS10}. The quantum query complexity of functions
nevertheless satisfies a weaker property called direct sum theorem,
meaning that computing $k$ instances requires at least $\Omega(k)$
times the number of queries necessary to solve one instance, but it is
unknown how the success probability decreases if less than $O(k)$
queries are used.

%In~\cite{Spa08}, \v{S}palek introduced the \emph{multiplicative} adversary method, generalizing previous ad-hoc methods for a set of problems. In particular, he showed that this method did not suffer from one particular limitation of the usual adversary method (which we will from now on call \emph{additive}), the fact that it cannot prove lower bounds for very small success probability.
%(this is not an issue for Boolean functions where the success probability is always at least $1/2$).
%However, he left unanswered the question of how multiplicative and additive methods relate in the case of high success probability.

\paragraph{Related work.}
We are not aware of any technique to
directly prove lower bounds for quantum state generation problems, and
the only few known lower bounds are based on reductions to computing
some functions. One particular example is a lower bound for the
already mentioned \IE{} problem, which consists in generating the
uniform superposition over the image of an injective function. The
best lower bound comes from a $\Omega(\sqrt[5]{N/ \log N})$ lower
bound for the \SE{} problem~\cite{Mid04}, which consists in deciding
whether two sets of size $N$ are equal or disjoint or, equivalently,
whether two injective functions over a domain of size $N$ have equal
or disjoint images. This problem reduces to \IE{} since by generating
the superposition over the image of the two functions, we can decide
whether they are equal or not using the SWAP-test. Therefore, this
implies the same $\Omega(\sqrt[5]{N/\log N})$ lower bound for \IE{}.
However, this lower bound is probably not tight, neither for \SE{},
whose best upper bound is $O(\sqrt[3]{N})$ due to the algorithm for
\collision{}~\cite{BHT97}, nor for \IE{}, whose best upper bound is
$O(\sqrt{N})$ due to an application of Grover's algorithm fro
\search~\cite{Gro96}. The question of the complexity of \IE{} has
first been raised by Shi~\cite{Shi02} in 2002 and has remained open
until the present work.

\paragraph{Our results.}
The chief technical innovation of this paper is an extension of both,
the additive and multiplicative adversary methods, to quantum state
generation \textbf{(Theorems~\ref{thm:additive-adversary}
  and~\ref{thm:MAM})}. To do so, we give a geometric interpretation of
the adversary methods which is reminiscent of the approach
of~\cite{Amb05, Spa08}, where this is done for classical problems. As
a by-product we give elementary and arguably more intuitive proofs of the
additive and multiplicative methods, contrasting with some rather
technical proofs \textit{e.g.} in~\cite{HLS07, Spa08}.

In order to compare the additive and multiplicative adversary bounds,
we introduce yet another flavor of adversary method
\textbf{(Theorem~\ref{thm:new-AAM})}, which we will call \emph{hybrid}
adversary method. Indeed, this method is a hybridization of the
additive and multiplicative methods that uses ``multiplicative''
arguments in an ``additive'' setup: it is equivalent to the additive
method for large success probability, but is also able to prove
non-trivial lower-bounds for small success probability, overcoming the
concern~\cite{Spa08} that the additive adversary method might fail in
this case. We show that for any problem, the hybrid adversary bound lies between
the additive and multiplicative adversary bounds
\textbf{(Theorem~\ref{thm:comparison-methods})}, answering \v{S}palek's open
question about the relative power of these
methods~\cite{Spa08}. By considering the \search{} problem for
exponentially small success probability, we also conclude that the
powers of the three methods are \emph{strictly} increasing, since the
corresponding lower bounds scale differently as a function of the
success probability in that regime
\textbf{(Theorem~\ref{thm:search})}.

We then extend the strong direct product theorem for the
multiplicative adversary bound~\cite{Spa08} to quantum state
generation problems \textbf{(Theorem~\ref{thm:SDPT})}. Since we have
clarified the relation between the additive and multiplicative
adversary methods, this also brings us closer to a similar theorem for
the additive adversary method. The most important consequence would be
for the quantum query complexity of functions, which would therefore
also satisfy a strong direct product theorem since the additive
adversary bound is tight in this case~\cite{LMRS10}. However, it
remains to prove some technical lemma about the multiplicative bound
to be able to conclude.

As it has been previously pointed out many interesting problems have strong symmetries~\cite{Amb05, ASW07, Spa08}. We show how studying these symmetries helps to address the two main difficulties of the usage the adversary method, namely, how to choose a good adversary matrix $\Gamma$
% \textbf{(Lemma~\ref{lem:symmetry-rho-gamma})}
and how to compute the spectral norm of $\Gamma_{x} - \Gamma$ \textbf{(Theorem~\ref{thm:adv-representation})}. Following the \emph{automorphism principle} of~\cite{HLS07}, we define the automorphism group $G$ of $\PP$, and its restrictions $G_x$, for any input $x$ to the oracle. We show how computing the norm of $\Gamma_{x} - \Gamma$ can be simplified to compute the norm of much smaller matrices that depend only on the irreps of $G$ and $G_x$. For problems with strong symmetries, these matrices typically have size at most $3\times 3$~\cite{Amb05, ASW07, Spa08}. We have therefore reduced the adversary method from an algebraic problem to the study of the representations of the automorphism group. 
 
%A similar approach has been used before to help designing $\aGamma$. We push this approach further to show how it can be explicitly used to compute the adversary bound itself. We show that it boils down to computing a quantity
%$\max_{l} \norm{\tilde{\Delta}_{x}^l}$,
%where the maximum is over irreps $l$ of $G_{x}$ and $\tilde\Delta^{l}_{x}$ is a matrix of size $m_{l}\times m_{l}$ (the multiplicity of $l$) depending only on irreps of $G$ and $G_{x}$ \textbf{(Theorem~\ref{thm:adv-representation})}. For many interesting problems with high symmetries, theses matrices are indeed a lot smaller than $\aGamma$ since they typically have size at most $2\times 2$~\cite{Amb05, ASW07, Spa08}. We therefore reduced the adversary method from an algebraic problem to the study of the representations of the automorphism group. 

Finally, we use our hybrid adversary method to prove a lower bound of $O(\sqrt{N})$ for the quantum query complexity of \IE{} \textbf{(Theorem~\ref{thm:index-erasure})}, which is tight due to the matching upper bound based on Grover's algorithm, therefore closing the open problem stated by Shi~\cite{Shi02}. To the best of our knowledge, this is the first lower bound directly proved for the query complexity of a quantum state generation problem. The lower bound is entirely based on the study of the representations of the symmetric group, a technique that might be fruitful for other problems having similar symmetries, such as the \SE{} problem~\cite{Mid04}, or in turn some stronger quantum state generation approaches to \GI.

%From these results, we can conclude that the multiplicative adversary method is a good candidate for a unified framework of lower bound techniques: it generalizes the usual adversary method which is (almost) tight for Boolean functions, as well as other ad-hoc lower bounds, and it can also be used to prove lower bounds on quantum state generation.

\section{Notations\label{sec:notations}}
%\subsection{Matrix norms}
In this paper, we will use different norms. We recall their definitions and some useful facts:
\begin{definition}
 For any matrix $A$, we use the following norms:
\begin{itemize*}
 \item Operator (or spectral) norm: $\norm{A} = \sup_{\ket{v}} \frac{\norm{A\ket{v}}}{\norm{\ket v}}$,
 \item Trace norm: $\tracenorm{A} = \Tr\sqrt{A^\dagger A}$,
 \item Frobenius norm: $\fnorm{A} = \sqrt{\Tr(A^\dagger A)}$.
\end{itemize*}
% \begin{align*}
% &\textrm{Operator norm:}&
% \norm{A} &= \sup_{\ket{v}} \frac{\norm{A\ket{v}}}{\norm{\ket v}},\\
% &\textrm{Trace norm:}&
% \tracenorm{A} &= \Tr\sqrt{A^\dagger A},\\
% &\textrm{Froebenius norm:}&
% \fnorm{A} &= \sqrt{\Tr(A^\dagger A)}.
% \end{align*}
\end{definition}
% \begin{definition}[Fidelity]
%  Let $\rho,\sigma$ be two mixed states (or density matrices) on a Hilbert space $\HH$, that is, trace 1 semi-definite positive matrices on $\HH$. The \emph{fidelity} $\FF(\rho,\sigma)$ between these states is:
% \begin{align*}
% 	\FF(\rho,\sigma) &= \Tr\sqrt{\sqrt\sigma\rho\sqrt\sigma}.
% \end{align*}
% \end{definition}
%Here are some useful properties.
\begin{lem}[H\"older's inequality]
\label{thm:Holder}
For any $A,B$, we have $\tracenorm{AB} \leq \fnorm{A} \cdot \fnorm{B}$.
\end{lem}
\begin{lem}
For any $A,B$, we have  $\Tr (AB)\leq\norm{A}\cdot\tracenorm{B}$.
\end{lem}
% \begin{lem}[Uhlmann's theorem \cite{Uhl76,Joz94}]
% \label{thm:Uhlmann}
% Let $\rho,\sigma$ be density matrices on $\HH$ and $\ket{\psi}\in\HH\otimes\HH'$ be any purification of $\rho$, that is, a vector such that $\Tr_{\HH'}\proj{\psi}=\rho$. Then, we have
% \begin{align*}
%  \FF(\rho,\sigma)&=\max_{\ket{\phi}}|\braket{\psi}{\phi}|,
% \end{align*}
% where the maximum is over purifications $\ket{\phi}\in\HH\otimes\HH'$ of $\sigma$.
% \end{lem}

In Section~\ref{sec:index-erasure} we will consider irreps of the symmetric group $S_{N}$, \textit{i.e., Young diagrams} and denote them by $\lambda_N, \lambda_N^{+},\dots$. Note that since a diagram $\lambda_N$ necessarily contains $N$ boxes, it is fully determined by its part $\lambda$ below the first row, as we know that its first row must contain $N-|\lambda|$ boxes, where $|\lambda|$ is the number of boxes below the first row. This will lighten the notations. The dimension of the space spanned by an irrep of the symmetric group can be easily computed:
\begin{lem}[Hook-length formula \cite{Sag01}]
For any Young diagram $\lambda$ corresponding to an irrep of $S_N$, the dimension of the space spanned by this irrep is: 
\begin{align*}
	d_{\lambda}^N &= \frac{N!}{\prod_{(i,j)\in\lambda}h_{i,j}},
\end{align*}
where $h_{i,j} = | \{ (i,j')\in \lambda_N : j' > j \} \cup \{ (i',j)\in \lambda_N  : i' \geq i \} |$.
\end{lem}

\section{Adversary methods: general concepts}
\subsection{Definition of the problem}
In this section, we describe elements which are common to all adversary methods.
The goal of these methods is to study the quantum query complexity of some problems in the bounded-error model when we have access to an oracle $O_f$ computing a function $f:\Sigma_{I}\mapsto\Sigma_{O}$. In this article, we will consider an oracle acting on two registers, the input register $\II$ and the output register $\OO$, as:
\begin{align*}
	\ket{x}_\II\ket{s}_\OO \stackrel{O_f}{\longrightarrow} \ket{x}_\II\ket{s\oplus f(x)}_\OO,
\end{align*}
where $x\in\Sigma_{I}$ and $s,f(x)\in\Sigma_{O}$. Note that it is also possible to consider other types of oracle, for example computing the value of the function into the phase instead of into another register, but these different models all lead to equivalent notions of query complexity (up to a constant).

We denote by $\Fset$ the set of all possible functions $f$ that can be encoded into the oracle.
We will consider three types of problems $\PP$, a classical one and two quantum ones:
\begin{description*}
	\item [Function] Given an oracle $O_f$, compute the classical output $\PP(f)$. The success probability of an algorithm $A$ solving $\PP$ is $\min_{f\in\Fset}\Pr[A(f)=\PP(f)]$, where $A(f)$ is the classical output of the algorithm on oracle $f$.
	\item [Coherent quantum state generation] Given an oracle $O_f$, generate a quantum state $\ket{\PP(f)} = \ket{\psi_f}$ in some target register $\TT$, and reset all other registers to a default state $\ket{\default}$. The success probability of an algorithm $A$ solving $\PP$ is given by $\min_{f\in\Fset}\norm{\Pi_{\ket{\PP(f)}}\otimes\Pi_{\ket{\default}}\ket{\psi_f^T}}^2$, where $\ket{\psi_f^T}$ is the final state of the algorithm and $\Pi_{\ket{\PP(f)}}$ and $\Pi_{\ket{\default}}$ are the projectors on the corresponding states.
 	\item [Non-coherent quantum state generation] Given an oracle $O_f$, generate a quantum state $\ket{\PP(f)} = \ket{\psi_f}$ in some target register $\TT$, while some $f$-dependent junk state may be generated in other registers. The success probability of an algorithm $A$ solving $\PP$ is given by $\min_{f\in\Fset}\norm{\Pi_{\ket{\PP(f)}}\ket{\psi_f^T}}^2$.
%  The success probability of an algorithm $A$ solving $\PP$ is defined as $\min_{f\in\Fset}\FF\!\left(A(f),\proj{\psi_f}\right)^{2}$, where $A(f)$ is the (mixed) state generated by the algorithm for oracle $f$.
\end{description*}
%We call \emph{quantum state generation} a problem in which the output is a quantum state, and \emph{classical} when the output is a bit-string.
%\comment{We should give more details here, in particular introduce the notation $\PP$, the target states $\ket{\psi_f}$, how we define the success probability,..}
Let us note that computing a function is a special case of non-coherent quantum state generation, where all states $\ket{\PP(f)}$ are computational basis states. Indeed, no coherence is needed since the state is in this case measured right after its generation. However, when the quantum state generation is used as a subroutine in a quantum algorithm for another problem, coherence is typically needed to allow interferences between different states. This is in particular the case for solving \SE{} via reduction to \IE{}, and similarly to solve $\GI{}$ via the quantum state generation approach, since coherence is required to implement the SWAP-test.

Without loss of generality we can consider the algorithm as being a
circuit $\CC$ consisting of a sequence of unitaries $U_0, \dots, U_T$
and oracle calls $O_f$ acting on the ``algorithm'' Hilbert space
$\AA$. Decomposing $\AA$ into three registers, the input register
$\II$ and output register $\OO$ for the oracle, as well as an
additional workspace register $\WW$, the circuit may be represented as in
Fig.~\ref{fig:circuit}.

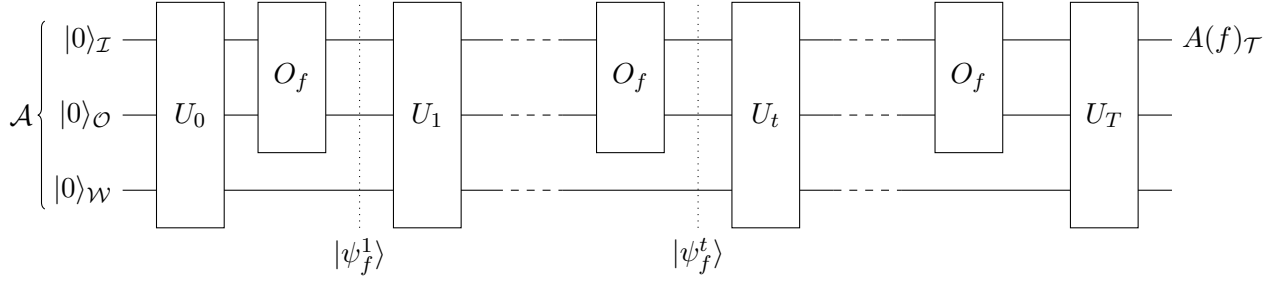
\begin{figure}
%\hspace*{-1cm}
\begin{tikzpicture}[xscale=.45, yscale=.5, decoration=brace]
\foreach \y in {1,3,5} {
 \draw (0,\y) -- (11,\y) ;  \draw[dashed] (11,\y) -- (13,\y) ; \draw (13,\y) -- (21,\y) ; \draw[dashed] (21,\y) -- (23,\y) ; \draw (23,\y) -- (31,\y) ;
  } 
  \node at (-3,3) {$\AA$};
  \draw[decorate] (-2.3,.5) -- (-2.3,5.5);
 \draw[fill=white] (1,0) rectangle (3,6) ; \node at (2,3) {$U_0$} ;
 %\draw[fill=white] (4,2) rectangle (6,8) ;  \node at (5,5) {$O_f$} ;
 \draw[fill=white] (4,2) rectangle (6,6) ;  \node at (5,4) {$O_f$} ;
 \draw[fill=white] (8,0) rectangle (10,6) ;  \node at (9,3) {$U_1$} ;
%\draw[fill=white] (14,2) rectangle (16,8) ;  \node at (15,5) {$O_f$} ;
\draw[fill=white] (14,2) rectangle (16,6) ;  \node at (15,4) {$O_f$} ;
 \draw[fill=white] (18,0) rectangle (20,6) ;  \node at (19,3) {$U_t$} ;
% \draw[fill=white] (24,2) rectangle (26,8) ;  \node at (25,5) {$O_f$} ;
  \draw[fill=white] (24,2) rectangle (26,6) ;  \node at (25,4) {$O_f$} ;
 \draw[fill=white] (28,0) rectangle (30,6) ;  \node at (29,3) {$U_T$} ;
 
% \node[left] at (0,7) {$\ket{f}_\FF$};
 \node[left] at (0,5) {$\ket{0}_\II$};
 \node[left] at (0,3) {$\ket{0}_\OO$};
 \node[left] at (0,1) {$\ket{0}_{\WW}$};
 %\node[right] at (31,7) {$\rho^T$} ;
 \node[right] at (31,5) {$A(f)_{\TT}$};
 
 \draw[dotted] (7,0) -- (7,6) ; \node[below] at (7,0) {$\ket{\psi_f^1}$};
% \draw[dotted] (9, 6.5) -- (9, 7.5) ; \node[above] at (9, 7.5) {$\rho^1$} ;
 \draw[dotted] (17,0) -- (17,6) ; \node[below] at (17,0) {$\ket{\psi_f^t}$};
%  \draw[dotted] (19, 6.5) -- (19, 7.5) ; \node[above] at (19, 7.5) {$\rho^t$} ;
\end{tikzpicture}

\caption{Schematic representation of a quantum algorithm that make use of an
  oracle $O_f$, an input register $\II$, an output register $\OO$, and a register $\WW$ for work space.\label{fig:circuit}}
\end{figure}

At the end of the circuit, a target register $\TT$ holds the output of the algorithm. In the classical case, this register is measured to obtain the classical output $A(f)$. In the quantum case, it holds the output state $A(f)$.

In both cases, for a fixed algorithm, we note $\ket{\psi_{f}^{t}}$ the state of the algorithm after the $t$-th query. The idea behind the adversary methods is to consider that $f$ is in fact an input to the oracle. We therefore introduce a function register $\FF$ holding this input, and define a \emph{super-oracle} $O$ acting on registers $\II\otimes\OO\otimes\FF$ as
\begin{align}
	\label{eqn:O}
	\ket{x}_\II\ket{s}_\OO\ket{f}_\FF \stackrel{O}{\longrightarrow} \ket{x}_\II\ket{s\oplus f(x)}_\OO\ket{f}_\FF.
\end{align}
We see that when the function register $\FF$ is in state $\ket{f}$, $O$ acts on $\II\otimes\OO$ just as $O_f$.
Suppose, just for the sake of analyzing the algorithm, that we prepare register $\FF$ in the state $\ket\delta = \inv{\sqrt{|\Fset|}}\sum_{f\in\Fset}\ket f$, the uniform superposition over all the elements of $\Fset$, and that we apply the same circuit as before, by replacing each call to $O_f$ by a call to $O$. Intuitively, each oracle call introduces more entanglement between this new register and the algorithm register. The state of this new circuit after the $t$-th query is (see Fig.~\ref{fig:circuit2}) 
\begin{align*}
  \ket{\Psi^t}&= \inv{\sqrt{|\Fset|}}\sum_{f\in\Fset}\ket{\psi_{f}^{t}}_{\AA}\ket{f}_{\FF}.
\end{align*}

\begin{figure}
\begin{tikzpicture}[xscale=.47, yscale=.5, decoration=brace]
\foreach \y in {1,3,5,7} {
 \draw (0,\y) -- (11,\y) ;  \draw[dashed] (11,\y) -- (13,\y) ; \draw (13,\y) -- (21,\y) ; \draw[dashed] (21,\y) -- (23,\y) ; \draw (23,\y) -- (31,\y) ;
  } 
    \node at (-3,3) {$\AA$};
  \draw[decorate] (-2.2,.5) -- (-2.2,5.5);
 \draw[fill=white] (1,0) rectangle (3,6) ; \node at (2,3) {$U_0$} ;
 \draw[fill=white] (4,2) rectangle (6,8) ;  \node at (5,5) {$O$} ;
 %\draw[fill=white] (4,2) rectangle (6,6) ;  \node at (5,4) {$O_f$} ;
 \draw[fill=white] (8,0) rectangle (10,6) ;  \node at (9,3) {$U_1$} ;
\draw[fill=white] (14,2) rectangle (16,8) ;  \node at (15,5) {$O$} ;
%\draw[fill=white] (14,2) rectangle (16,6) ;  \node at (15,4) {$O_f$} ;
 \draw[fill=white] (18,0) rectangle (20,6) ;  \node at (19,3) {$U_t$} ;
\draw[fill=white] (24,2) rectangle (26,8) ;  \node at (25,5) {$O$} ;
 % \draw[fill=white] (24,2) rectangle (26,6) ;  \node at (25,4) {$O_f$} ;
 \draw[fill=white] (28,0) rectangle (30,6) ;  \node at (29,3) {$U_T$} ;
 
\node[left] at (0,7) {$\ket{\delta}_\FF$};
 \node[left] at (0,5) {$\ket{0}_\II$};
 \node[left] at (0,3) {$\ket{0}_\OO$};
 \node[left] at (0,1) {$\ket{0}_{\WW}$};
 %\node[right] at (31,7) {$\rho^T$} ;
 \node[right] at (31,5) {$A(f)_{\TT}$};
 \node[right] at (31,7) {$\rho^{T}$};
 
 \draw[dotted] (7,0) -- (7,8) ; \node[below] at (7,0) {$\ket{\Psi^1}$};
 \draw[dotted] (9, 6.5) -- (9, 7.5) ; \node[above] at (9, 7.5) {$\rho^1$} ;
 \draw[dotted] (17,0) -- (17,8) ; \node[below] at (17,0) {$\ket{\Psi_t}$};
  \draw[dotted] (19, 6.5) -- (19, 7.5) ; \node[above] at (19, 7.5) {$\rho^t$} ;

\end{tikzpicture}
\caption{Schematic representation of a quantum algorithm that makes use of 
an oracle $O_f$, an input register $\II$, an output register $\OO$, a register $\WW$ for work space, and a virtual register
$\FF$ holding the input of the problem.\label{fig:circuit2}}
\end{figure}
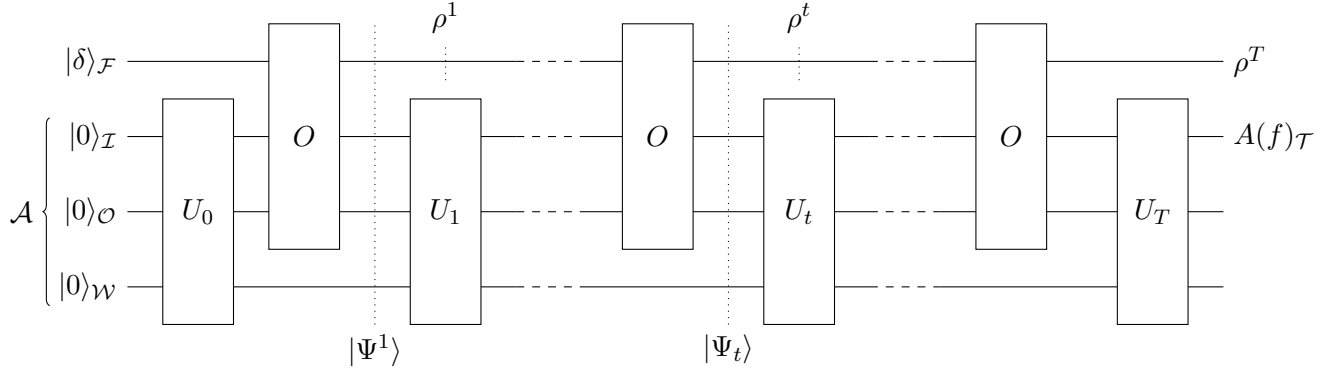

Note that only oracle calls can modify the state of the function register $\FF$, since all other gates only affect the algorithm register $\AA=\II\otimes\OO\otimes\WW$. The general idea of all adversary methods is to study the evolution of the algorithm by looking at the reduced state of the input register,
\begin{align*}
\rho^{t} = \Tr_{\AA} \proj{\Psi^t} = \inv{|\Fset|} \sum_{f,f'\in\Fset} \braket{\psi_{f'}^{t}}{\psi_{f}^{t}} \ketbra{f}{f'}.
\end{align*}

The algorithm starts with the state $\rho^{0} = \proj{\delta}$ and ends in a state $\rho^{T}$.

\subsection{Adversary matrices and progress function}

%The success of the algorithm is defined by $F(\rho^{T},\target)$ where $\target$ the target state with no error.
%\comment{Not exactly because of the junk, we should fix this.}
The adversary method studies how fast $\rho^{t}$ can change from $\rho^{0}$ to $\rho^{T}$. We introduce a progress function in order to do so.

\begin{definition}[Adversary matrix]
 An adversary matrix $\Gamma$ is a Hermitian matrix such that $\Gamma\ket{\delta}=\ket\delta$. An additive adversary matrix also satisfies $-\I\preceq\Gamma\preceq\I$ (i.e., $\norm{\Gamma}=1$), while a multiplicative adversary matrix satisfies $\Gamma\succeq\I$. In both cases, the progress function is defined as $W^t=\Tr\left[\Gamma\rho^t\right]$.
\end{definition}

We will also use a matrix $\Gamma_x$ derived from the adversary matrix $\Gamma$ and defined as follows (for both the additive and the multiplicative case).
\begin{definition}[$\Gamma_x, D_x$]
For any adversary matrix $\Gamma$, let $\Gamma_x=\Gamma\circ D_x$, where $\circ$ denotes the Hadamard (element-wise) product and $D_{x}$ is the (0-1)-matrix $D_{x} = \sum_{f,f'} \delta_{f(x),f'(x)}\ketbra{f'}{f}$ where $\delta$ denotes the Kronecker's delta.
\end{definition}
We will show that the Hadamard product is closely related to oracle calls: when the input register is in the state $\ket{x}$, the oracle calls acts on the function register as the Hadamard product with $D_x$. It is easy to check that this Hadamard product is a CP-map.
\begin{fact}\label{fact:CP-map}
 The map $\gamma \mapsto \gamma \circ D_{x}$ is a CP-map and $\gamma \circ D_{x} = \sum_{y} \Pi_y^x \gamma \Pi_y^x$ with $\Pi_y^x = \sum_{f: f(x)=y} \proj{f}$. 
\end{fact}

% \subsection{Effect of oracle calls}
The basic idea of all adversary methods is to bound how much the value of the progress function can change by one oracle call. To study the action of one oracle call, we isolate the registers $\II$ and $\OO$ holding the input and output of the oracle from the rest of the algorithm register.
Without loss of generality, we may assume that for any oracle call, the output register $\OO$ is in the state $\ket{0}_\OO$ (\emph{computing} oracle call) or $\ket{f(x)}_\OO$ (\emph{uncomputing} oracle call). Indeed, an oracle call for any other state $\ket{s}_\OO$ may be simulated by one computing oracle call, $\OO(\log N)$ XOR gates and one uncomputing oracle call. Therefore, this assumption only increases the query complexity by a factor at most 2.

Let us consider the action of the $(t+1)$-th oracle call, which we assume to be of \emph{computing} type (uncomputing oracle calls are treated similarly).
Just before the $(t+1)$-th oracle call, the state can be written as:
\begin{align*}
	\ket{\Psi^{t}} = \inv{\sqrt{|\Fset|}} \sum_{x,f} \ket{\psi^t_{f,x}}_{\WW}\ket{x}_{\II}\ket{0}_{\OO}\ket{f}_{\FF},
\end{align*}
with $\ket{\psi^{t}_{f,x}}$ being non-normalized states.
Let us consider the reduced density matrix
\begin{align}
	\tilde{\rho}^{t} = \Tr_{\WW} \proj{\Psi^{t}} = \inv{|\Fset|} \sum_{f,f',x,x'} \braket{\psi^t_{f,x}}{\psi^t_{f',x'}} \ketbra{x'}{x} \otimes \proj{0} \otimes \ketbra{f'}{f} \label{eqn:defRhoTildeT},
\end{align}
and note that $\rho^t=\Tr_{\II\OO}\left[\tilde\rho^{t}\right]$.

\begin{lem}\label{lemma:progress-oracle}
 Let the $t$-th oracle call be of computing-type. Then,   $W^{t}=\Tr\left[\Upsilon\tilde\rho^{t}\right]$ and $W^{t+1}=\Tr\left[\Upsilon'\tilde\rho^{t}\right]$, where 
\begin{align}
\Upsilon &= \sum_{x} \proj{x} \otimes  \sum_{y} \proj{y} \otimes \Gamma = \bigoplus_{x,y} \Gamma, \label{eqn:defG} \\
\Upsilon' &= \sum_{x} \proj{x} \otimes  \sum_{y} \proj{y} \otimes \Gamma_{x} = \bigoplus_{x,y} \Gamma_{x} \label{eqn:defGPrime}.
\end{align}
\end{lem}
Note that for uncomputing oracle calls, it suffices to swap the roles of $\rho^t$ and $\rho^{t+1}$.

\begin{proof}
From the definition of $\Upsilon$ and the fact that $\rho^t=\Tr_{\II\OO}\left[\tilde\rho^{t}\right]$, we immediately have that $W^{t}=\Tr\left[\Gamma\rho^{t}\right]=\Tr\left[\Upsilon\tilde\rho^{t}\right]$. Let us now consider what happens after one oracle call.
An oracle call acts on the registers $\II\otimes\OO\otimes\FF$ as the operator
\begin{align*}
 O = \sum_{x} \proj{x} \sum_{f,s} \ketbra{f(x)\oplus s}{s} \otimes \proj{f}.
\end{align*}
Before a computing oracle call, the output register $\OO$ is in the state $\ket{0}$, as in eq.~(\ref{eqn:defRhoTildeT}). Therefore, the state $\tilde\rho^{t+1} = O \tilde\rho^{t}O^{\dagger}$ just after the $(t+1)$-th oracle call is
\begin{align*}
\tilde\rho^{t+1} &= \inv{|\Fset|} \sum_{f,f',x,x'} \braket{\psi^t_{f,x}}{\psi^t_{f',x'}} \ketbra{x'}{x} \otimes \ketbra{f'(x')}{f(x)} \otimes \ketbra{f'}{f}
\end{align*}
and
\begin{align}
\rho^{t+1} = \Tr_{\II\OO}\left[\tilde\rho^{t+1}\right] = \sum_{x} \rho_{x}^t \circ D_{x}, \label{eqn:RhoTPlusOne}
\end{align}
where
\begin{align}
\label{eqn:Rhotx}
 \rho_{x}^t=\inv{|\Fset|} \sum_{f,f'} \braket{\psi^t_{f,x}}{\psi^t_{f',x}} \ketbra{f'}{f}
\end{align}
Combining eqs.~(\ref{eqn:defRhoTildeT}) and~(\ref{eqn:defGPrime}) we have:
\begin{align*}
\Tr\left[\Upsilon'\tilde\rho^{t}\right] &= \inv{|\Fset|} \sum_{f,f',x,x'} \Tr\left[\braket{\psi^t_{f,x}}{\psi^t_{f',x'}} \ketbra{x'}{x} \otimes \proj{0} \otimes \Gamma_{x} \ketbra{f'}{f}\right] \\
&=  \inv{|\Fset|} \sum_x \Tr\left[ \Gamma_{x} \sum_{f,f'} \braket{\psi^t_{f,x}}{\psi^t_{f',x}} \ketbra{f'}{f} \right] \\
&= \sum_x \Tr\left[ \Gamma_{x} \rho_{x}^{t} \right] \quad \text{by eq.~(\ref{eqn:Rhotx})}\\
&= \sum_{x} \Tr\left[ (\Gamma \circ D_{x}) \rho_{x}^{t}  \right]\\
&= \sum_{x} \Tr\left[ \Gamma ( \rho_{x}^{t} \circ D_{x}) \right]
\textrm{\quad using Fact~\ref{fact:CP-map} and }\Tr(AB)=\Tr(BA)  \\
&=\Tr\left[\Gamma\rho^{t+1}\right] \quad \textrm{by eq.~(\ref{eqn:RhoTPlusOne}).}
\end{align*}
\end{proof}

\section{The different adversary methods}
\subsection{Additive adversary method}

\emph{Additive adversary} should be understood as \emph{adversary with negative weights} as defined in \cite{HLS07}. To differentiate between the different methods, we will from now on denote additive adversary matrices by $\aGamma$ and multiplicative adversary matrices by $\Gamma$. For the statement of the theorem, we will also need the following notions.
\begin{definition}[$\rho^{\odot}$, junk matrix]
For a quantum state generation problem $\PP$ such that $\ket{\PP(f)}=\ket{\psi_f}$, we denote by $\target$ the \emph{target} state
$
\target=\inv{|\Fset|}\sum_{f,f'\in\Fset}\braket{\psi_f}{\psi_{f'}}\ketbra{f'}{f}.
$
In the non-coherent case, we call junk matrix any Gram matrix $M$ of size $|\Fset|\times|\Fset|$ such that $M_{ij}=\braket{v_i}{v_j}$, where $\{v_i:i\in[|\Fset|]\}$ is a set of unit vectors (or, equivalently, any semi-definite matrix $M$ such that $M_{ii}=1$ for any $i\in[|\Fset|]$). In the coherent case, we call junk matrix the all-1 matrix of size $|\Fset|\times|\Fset|$.
\end{definition}

\begin{thm}[Additive adversary method \cite{HLS07}]\label{thm:additive-adversary}
Consider a quantum algorithm solving $\PP$ with success probability at least $1-\varepsilon$, and let $\aGamma$ be an additive adversary matrix such that $\Tr\left[\aGamma(\target \circ M)\right]=0$ for any junk matrix $M$. Then,
%\begin{enumerate*}
%	\item $\aW^0 = 1$,
%	\item $|\aW^{t+1}-\aW^t| \leq \max_{x}{\norm{\aGamma_{x}-\aGamma}}$,
%% where $\aGamma_{x} = \aGamma \circ D_{x}$ and $D_{x} = \sum_{\substack{f,f'\in\Fset\\f(x) = f'(x)}}\ketbra{f'}{f}$,
%	\item $\aW^{T} \leq C(\varepsilon)$,
%where $C(\varepsilon)=\varepsilon+2\sqrt{\varepsilon(1-\varepsilon)}$.
%\end{enumerate*}
%Therefore,
\begin{align*}
	Q_{\varepsilon}(\PP) \geq \frac{1-C(\varepsilon)}{\max_{x}\norm{\aGamma_{x}-\aGamma}} \quad\text{where}\quad C(\varepsilon)=\varepsilon+2\sqrt{\varepsilon(1-\varepsilon)}
\end{align*}
\end{thm}
\begin{proof}
By definition of $\aGamma$, the initial value of the progress function is $\aW^0 = 1$. 
%Item 1 is immediate by assumption on $\aGamma$.
We now bound the decrease of the the progress function for each query. We have from Lemma~\ref{lemma:progress-oracle}
\begin{align*}
 \abs{\aW^{t+1}-\aW^t}=\abs{\Tr[(\aUpsilon'-\aUpsilon)\tilde{\rho}^t]}\leq\norm{\aUpsilon'-\aUpsilon}=\max_{x}{\norm{\aGamma_{x}-\aGamma}}.
\end{align*}

To conclude, we need to upper-bound the value of the progress function at the end of the algorythm. Let us prove that $\aW^T \leq C(\varepsilon)$. Let $\ket{\psi_f}$ be the state to be generated when the input is $f$ (in particular, for a classical problem this will just be a computational basis state encoding the output of the classical problem). The final state is:
 \begin{align*}
 	\ket{\Psi^T}=\frac{1}{\sqrt{|\Fset|}}\sum_{f\in\Fset}\left[ \sqrt{1-\varepsilon_f}\ket{\psi_f,\junk_f}+\sqrt{\varepsilon_f}\ket{\mathrm{err}_f}\right]\ket{f},
 \end{align*}
where $\ket{\junk_f}$ is the default state $\ket{\default}$ for a coherent quantum state generation problem, and any state otherwise.
 Since the algorithm has success probability $1-\varepsilon$, we have $0 \leq \varepsilon_f \leq \varepsilon, \forall f$ and the final state can be rewritten as:
\begin{align*}
 \ket{\Psi^T}=\frac{1}{\sqrt{|\Fset|}}\sum_{f\in\Fset}\left[ \sqrt{1-\varepsilon}\ket{\psi_f,\junk_f}+\sqrt{\varepsilon}\ket{\error_f}\right]\ket{f},
\end{align*}
where $\ket{\error_f}$ is the (non-normalized) vector $\frac{\sqrt{1-\varepsilon_f}-\sqrt{1-\varepsilon}}{\sqrt\varepsilon}\ket{\psi_f,\junk_f}+\sqrt{\frac{\varepsilon_f}{\varepsilon}}\ket{\mathrm{err}_f}$.

Tracing over everything but the last register, we have
\begin{align*}
 \rho^T=(1-\varepsilon)\left(\target \circ M_\junk \right) +\varepsilon \tau +\sqrt{\varepsilon(1-\varepsilon)} (\sigma+\sigma^\dagger),
\end{align*}
where
\begin{align*}
 M_\junk &=\sum_{f,f'\in\Fset}\braket{\junk_f}{\junk_f'}\ketbra{f'}{f}, \\
 \tau &=\frac{1}{|\Fset|}\sum_{f,f'\in\Fset}\braket{\error_f}{\error_f'}\ketbra{f'}{f}, \\
\sigma&=\frac{1}{|\Fset|}\sum_{f,f'\in\Fset}\braket{\psi_f,\junk_f}{\error_f'}\ketbra{f'}{f}.
\end{align*}
By assumption on $\aGamma$, we have $\Tr\left[\aGamma(\target \circ M_\junk)\right]=0$ and $\Tr\left[\aGamma A\right]\leq\tracenorm{A}$ for any operator $A$, so that
\begin{align*}
 W^T&=(1-\varepsilon)\Tr\left[\aGamma(\target \circ M_\junk)\right]
+\varepsilon\Tr\left[\aGamma\tau\right]
+\sqrt{\varepsilon(1-\varepsilon)}\Tr\left[\aGamma(\sigma+\sigma^\dagger)\right]
\\
&\leq\varepsilon\tracenorm{\tau}+\sqrt{\varepsilon(1-\varepsilon)}\tracenorm{\sigma+\sigma^\dagger}.
\end{align*}
It remains to show that $\tracenorm{\tau}\leq 1$ and $\tracenorm{\sigma+\sigma^\dagger}\leq 2$.
Let us define the following matrices.
\begin{align*}
 A&=\frac{1}{\sqrt{|\Fset|}}\sum_{f\in\Fset}\ketbra{\psi_f,\junk_f}{f},&
 B&=\frac{1}{\sqrt{|\Fset|}}\sum_{f\in\Fset}\ketbra{\error_f}{f}.
\end{align*}
Then, we have $\sigma=(A^\dagger B)^t$ and therefore $\tracenorm{\sigma+\sigma^\dagger} \leq 2 \tracenorm{\sigma}  = 2\tracenorm{A^\dagger B} \leq 2\fnorm{A}\cdot\fnorm{B}  \leq 2$,
where we have used H\"older's inequality (Lemma~\ref{thm:Holder}) and the fact that $\fnorm{A}=1$ since $\ket{\psi_f,\junk_f}$ is normalized, and $\fnorm{B} \leq 1$ and $ \braket{\error_f}{\error_f} = \inv{\varepsilon}\left(2-\varepsilon-2\sqrt{1-\varepsilon}\sqrt{1-\varepsilon_f} \right) \leq 1$ for $\varepsilon_f \leq \varepsilon$. Similarly, we have $\tau=(B^\dagger B)^t$ and therefore $\tracenorm{\tau}\leq\fnorm{B}^2\leq 1$.
\end{proof}

%%%%%%%%%%%%%%%%%%%%%%
%\subsection{How to choose the adversary matrix}
%\subsection{Special case of classical problems}
%%%%%%%%%%%%%%%%%%%%%%
For classical problems, we now prove that our method generalizes \cite{HLS07}. Indeed, our condition on the adversary matrix is different, which allows us to also deal with quantum problems.
However, for classical problems, the following lemma shows that the usual condition implies our modified condition. Let $\PP(f)$ be the function to be computed.
\begin{lem}
  $\Tr\left[\aGamma(\target \circ M)\right]=0$ for any matrix $M$
if and only if $\aGamma_{ff'}=0$ for any $f,f'$ such that $\PP(f)=\PP(f')$.
\end{lem}
\begin{proof}
 Let $\aGamma$ be such that $\Tr\left[\aGamma(\target \circ M)\right]=0$ for any matrix $M$, and $\bar{f},\bar{f'}$ be such that $\PP(\bar{f})=\PP(\bar{f'})$. Choosing $M$ such that $M_{\bar{f}\bar{f'}}=1$ and $M_{ff'}=0$ for any other element, we have $\target \circ M=\frac{1}{|\Fset|}M$ and therefore $\aGamma_{\bar{f}\bar{f'}}=0$.

For the other direction, we obtain for any matrix $M$
\begin{align*}
 \Tr\left[\aGamma(\target \circ M)\right]=\inv{|\Fset|} \sum_{f,f'\in\Fset}\aGamma_{ff'}\braket{\PP(f)}{\PP(f')}M_{ff'}=0
\end{align*}
since $\aGamma_{ff'}=0$ whenever $\PP(f)=\PP(f')$, and $\braket{\PP(f)}{\PP(f')}=0$ whenever $\PP(f)\neq \PP(f')$.
\end{proof}

\subsection{Hybrid adversary method}
The original adversary method can only prove a lower bound when $C(\varepsilon)<1$, that is, when the success probability $1-\varepsilon>\frac{4}{5}$. For smaller success probability, we need to prove a stronger bound on the final value of the progress function $\aW^T$. Inspired by the multiplicative adversary method~\cite{Spa08}, we prove the following \emph{hybrid} adversary bound.
% \begin{lem}\label{lem:success-proba}
%  Let $\rho^T=\Tr_\AA\proj{\Psi^T}$, where $\ket{\Psi^T}$ is the final state of an algorithm for $\PP$. Then, the success probability $p$ of the algorithm satisfies
% $p\leq\max_{M} \left[\FF(\rho^T,\target\circ M)\right]^2$,
% where the maximum is over all $F \times F$ normalized Gram matrices $M$.
% \end{lem}
% \begin{proof}
%  Recall from the proof of Theorem~\ref{thm:additive-adversary} that $\ket{\Psi^T}$ can be written as
% \begin{align*}
%  \ket{\Psi^T}&=\frac{1}{\sqrt{|\Fset|}}\sum_{f\in\Fset}\left[ \sqrt{p}\ket{\psi_f,\junk_f}+\sqrt{1-p}\ket{\error_f}\right]\ket{f},
% \end{align*}
%  and let us consider the following state
% \begin{align*}
%  \ket{\Psi}&=\frac{1}{\sqrt{|\Fset|}}\sum_{f\in\Fset}\ket{\psi_f,\junk_f}\ket{f}.
% \end{align*}
% Since this state is a purification of $\target \circ M_\junk$, and $|\braket{\Psi}{\Psi^T}| \geq \sqrt{p}$, we have from Uhlmann's theorem (Lemma~\ref{thm:Uhlmann})
% $
%  \FF(\rho^T,\target\circ M_\junk)\geq \sqrt{p}.
% $
% \end{proof}
% 
% This motivates the definition of the quantity $\eta(\rho)$ for any $F\times F$ density matrix by:
% \begin{align*}
% 	\label{eqn:EtaOfRho}
% 	\eta(\rho) &= \max_{M} \left[\FF(\rho,\target\circ M)\right]^2,
% \end{align*}
% where the maximum is over all $F \times F$ normalized Gram matrices $M$.

\begin{thm}[Hybrid adversary method]\label{thm:new-AAM}
Consider a quantum algorithm solving $\PP$ with success at least $1-\varepsilon$. Let $\aGamma$ be any additive adversary matrix, $V_\bad$ be the direct sum of eigenspaces of $\aGamma$ with eigenvalue strictly larger than $\alambda<1$, and assume that
$\Tr\left[\Pi_\bad(\target \circ M)\right]\leq\eta$ for any junk matrix $M$, where $\Pi_\bad$ is the projector on $V_\bad$, and $0\leq\eta\leq 1-\varepsilon$.
We have
%\begin{enumerate*}
%	\item $\aW^0 = 1$,
%	\item $|\aW^{t+1}-\aW^t| \leq \max_{x}{\norm{\aGamma_{x}-\aGamma}}$,
%% where $\aGamma_{x} = \aGamma \circ D_{x}$ and $D_{x} = \sum_{\substack{f,f'\in\Fset\\f(x) = f'(x)}}\ketbra{f'}{f}$,
%	\item $\aW^{T} \leq 1-\tilde K(\aGamma,\alambda,\varepsilon)$, where $\tilde{K}(\aGamma,\alambda,\varepsilon)=(1-\alambda)(\sqrt{1-\varepsilon}-\sqrt{\eta})^2$.
%\end{enumerate*}
%Therefore,
\begin{align*}
	Q_{\varepsilon}(\PP) \geq \frac{\tilde K(\aGamma,\alambda,\varepsilon)}{\max_{x}\norm{\aGamma_{x}-\aGamma}} \quad \text{where} \quad \tilde{K}(\aGamma,\alambda,\varepsilon)=(1-\alambda)(\sqrt{1-\varepsilon}-\sqrt{\eta})^2.
\end{align*}
\end{thm}

\begin{proof}
The initial value of the progress function and the bound on the amount of change between two queries are the same as the additive adversary method, so we only need to prove that $\aW^{T} \leq 1-\tilde K(\aGamma,\alambda,\varepsilon)$. Recall that by assumption, $\ket{\Psi^T}$ can be written 
\begin{align*}
 \ket{\Psi^T}&=\frac{1}{\sqrt{|\Fset|}}\sum_{f\in\Fset}\left[ \sqrt{1-\varepsilon}\ket{\psi_f,\junk_f}+\sqrt{\varepsilon}\ket{\error_f}\right]\ket{f}.
\end{align*}
The state $\ket{\Psi} =\frac{1}{\sqrt{|\Fset|}}\sum_{f\in\Fset}\ket{\psi_f,\junk_f}\ket{f}$ satisfies $|\braket{\Psi}{\Psi^T}|\geq \sqrt{1-\varepsilon}$, and $\Tr_\AA\proj{\Psi}=\target\circ M_\junk$. Let $\beta=\norm{\Pi_\good\ket{\Psi^T}}^2$, $\ket{\Psi_\good}=\Pi_\good\ket{\Psi^T}/\sqrt{\beta}$ and $\ket{\Psi_\bad}=\Pi_\bad\ket{\Psi^T}/\sqrt{1-\beta}$, so that
% \begin{align*}
%  \ket{\Psi^T}&=\sqrt{\beta}\ket{\Psi_\good}+\sqrt{1-\beta}\ket{\Psi_\bad}.
% \end{align*}
% Let us also define $\rho_\bad=\Tr_\AA\proj{\Psi_\bad}$ and $\rho_\good=\Tr_\AA\proj{\Psi_\good}$.
% We have
\begin{align*}
 \sqrt{1-\varepsilon}&\leq|\braket{\Psi}{\Psi^T}|
=\sqrt{\beta} \ |\braket{\Psi}{\Psi_\good}|+\sqrt{1-\beta}\ |\braket{\Psi}{\Psi_\bad}|\\
&\leq\sqrt{\beta} \ \norm{\Pi_\good\ket{\Psi}}+\sqrt{1-\beta}\ \norm{\Pi_\bad\ket{\Psi}}\\
&\leq\sqrt{\beta}+\sqrt{1-\beta}\ \sqrt{\Tr\left[\Pi_\bad(\target \circ M_\junk)\right]}\\
&\leq\sqrt{\beta}+\sqrt{\eta}.
\end{align*}
% \comment{This now implies Robert's bound (see page 7 of~\cite{spalek08}), since taking the square of this expression, we have $1-\varepsilon\leq\eta+\beta+2\sqrt{\eta\beta}\leq\eta+3\sqrt{\beta}$, where we have used $\beta\leq 1$ and $\eta\leq 1$.}
Since $\eta\leq 1-\varepsilon$, we obtain that $\beta\geq(\sqrt{1-\varepsilon}-\sqrt{\eta})^2$. We are now ready to bound $\aW^{T}=\Tr(\aGamma\rho^T)$, where $\rho^T =\beta\rho_\good+(1-\beta)\rho_\bad+\sqrt{\beta(1-\beta)}\left[\Tr_A(\ket{\Psi_\good}\bra{\Psi_\bad})+\Tr_A(\ket{\Psi_\bad}\bra{\Psi_\good})\right]$.

Since $\Tr(\aGamma\rho_\good)\leq\alambda$, $\Tr(\aGamma\rho_\bad)\leq 1$, and the off-diagonal terms are zero, we have
\begin{align}
 \aW^{T}&=\beta\ \Tr(\aGamma\rho_\good)+(1-\beta)\ \Tr(\aGamma\rho_\bad)\label{eq:progress-good-bad}\\
&\leq 1-(1-\alambda)\beta
\leq 1-(1-\alambda)(\sqrt{1-\varepsilon}-\sqrt{\eta})^2.
\end{align}
\end{proof}

%%%%%%%%%%%%%%%%%%%%%%%%
%\subsection{How to choose the adversary matrix}
%\subsection{Special case of classical problems}
%%%%%%%%%%%%%%%%%%%%%%%%

For classical problems, we can use the following lemma:
\begin{lem}\label{lem:new-additive-classical}
 Let $\Pi_\bad$ be the projector on $V_\bad$, $\Pi_z=\sum_{\PP(f)=z}\proj{f}$, and assume that $\norm{\Pi_z\Pi_\bad}^2\leq\eta$ for any $z$. Then,
$\Tr\left[\Pi_\bad(\target \circ M)\right]\leq\eta$ for any junk matrix $M$.
\end{lem}
\begin{proof}
For any junk matrix $M$, let us define the following purification of $\target\circ M$,
\begin{align*}
 \ket{\psitarget_M}=\frac{1}{\sqrt{|\Fset|}}\sum_f\ket{\PP(f)}\ket{M_f}\ket{f},
\end{align*}
where $\ket{M_f}$ are normalized states such that $\braket{M_f}{M_{f'}}=\bra{f}M\ket{f'}$. Let us also consider the operator
$
 P=\sum_z \proj{z}\otimes \Pi_z.
$
Then, we have $P\ket{\psitarget_M}=\ket{\psitarget_M}$, so that
\begin{align*}
 \Tr\left[\Pi_\bad(\target \circ M)\right] =\norm{\Pi_\bad\ket{\psitarget_M}}^2
=\norm{\Pi_\bad P\ket{\psitarget_M}}^2
 \leq\norm{\Pi_\bad P}^2=\max_z\norm{\Pi_\bad \Pi_z}^2\leq\eta.
\end{align*}
\end{proof}

\subsection{Multiplicative adversary method}
\begin{thm}[Multiplicative adversary method \cite{Spa08}]\label{thm:MAM}
Consider a quantum algorithm solving $\PP$ with success at least $1-\varepsilon$. Let $\Gamma$ be any multiplicative adversary matrix, $V_\bad$ be the direct sum of eigenspaces of $\Gamma$ with eigenvalue strictly smaller than $\lambda>1$, and assume that
$\Tr\left[\Pi_\bad(\target \circ M)\right]\leq\eta$ for any junk matrix $M$, where $\Pi_\bad$ is the projector on $V_\bad$, and $0\leq\eta\leq 1-\varepsilon$.
We have
%\begin{enumerate*}
%	\item $W^0 = 1$
%	\item $\frac{W^{t+1}}{W^t} \leq \max \left\{\norm{\Gamma_{x}^{1/2}\Gamma^{-1/2}}^2
%,\norm{\Gamma^{1/2}\Gamma_{x}^{-1/2}}^2:\forall x\in\II\right\}$
%	\item $W^{T} \geq K(\Gamma,\lambda,\varepsilon)$, where $K(\Gamma,\lambda,\varepsilon)=1+(\lambda-1)(\sqrt{1-\varepsilon}-\sqrt{\eta})^2$.
%\end{enumerate*}
%Therefore,
\begin{align*}
	Q_{\varepsilon}(\PP) \geq \frac{\log K(\Gamma,\lambda,\varepsilon)}{\log \max \left\{\norm{\Gamma_{x}^{1/2}\Gamma^{-1/2}}^2
,\norm{\Gamma^{1/2}\Gamma_{x}^{-1/2}}^2:\forall x\in\II\right\} },% \text{with} K(\Gamma,\lambda,\varepsilon)=1+(\lambda-1)(\sqrt{1-\varepsilon}-\sqrt{\eta})^2.
\end{align*}
where $K(\Gamma,\lambda,\varepsilon)=1+(\lambda-1)(\sqrt{1-\varepsilon}-\sqrt{\eta})^2$.
\end{thm}

% In lemmas (\ref{thmc:classicalMAM}) and (\ref{thmc:quantumMAM}) we will define more precisely $K(\Gamma,\zeta)$ and in corollary (\ref{thmc:symmetrization}) we will examine how to use the symmetries of the problem in order to compute efficiently $\norm{\Gamma_{x}^{1/2}\Gamma^{-1/2}}^2$ and $\norm{\Gamma^{1/2}\Gamma_{x}^{-1/2}}^2$.

\begin{proof}
As done in the previous proof, the initial value of the progress function is $W^0=1$. 

In this case we do not bound the difference of the progress function between two queries, but its quotient. From Fact~\ref{fact:CP-map}, we note that $\Upsilon$ and $\Upsilon'$ are definite-positive. Then, using Lemma~\ref{lemma:progress-oracle}, we have
\begin{align*}
\frac{W^{t+1}}{W^{t}} &= \frac{\Tr\left[\Upsilon'\tilde\rho^{t}\right]}{\Tr\left[\Upsilon\tilde\rho^{t}\right]} = \frac{\Tr\left[\Upsilon'^{1/2}\Upsilon^{-1/2}\Upsilon^{1/2}\tilde\rho^{t}\Upsilon^{1/2}\Upsilon^{-1/2}\Upsilon'^{1/2}\right]}{\Tr\left[\Upsilon^{1/2}\tilde\rho^{t}\Upsilon^{1/2}\right]} \\
& \leq \norm{\Upsilon'^{1/2}\Upsilon^{-1/2}}^{2} = \norm{\bigoplus_{x,y}\Gamma_{x}^{1/2}\Gamma^{-1/2}}^{2} = \max_{x} \norm{\Gamma_{x}^{1/2}\Gamma^{-1/2}}^{2},
\end{align*}
%where we have used Lemma~\ref{lemma:progress-oracle}.
If the $(t+1)$-th oracle call is of \emph{uncomputing} type, we similarly obtain 
$\frac{W^{t+1}}{W^{t}}  \leq \max_{x} \norm{\Gamma^{1/2}\Gamma_{x}^{-1/2}}^{2}$.
%Together with the previous equation, this implies Item 2 of Theorem~\ref{thm:MAM}.

The proof of the upper bound of $W^T$ is similar to the one in Theorem~\ref{thm:new-AAM} up to eq.~(\ref{eq:progress-good-bad}), where we now have $W^{T} =\beta\Tr(\Gamma\rho_\good)+(1-\beta)\Tr(\Gamma\rho_\bad) \geq 1+(\lambda-1)\beta  \geq 1+(\lambda-1)(\sqrt{1-\varepsilon}-\sqrt{\eta})^2$.

The lower-bound on the query complexity is a consequence of 
\begin{align*}\left(\max \left\{\norm{\Gamma_{x}^{1/2}\Gamma^{-1/2}}^2
,\norm{\Gamma^{1/2}\Gamma_{x}^{-1/2}}^2:\forall x\in\II\right\}\right)^T \geq K(\Gamma,\lambda,\varepsilon).
\end{align*}
\end{proof}
Note that since the condition on the adversary matrix is very similar as for the hybrid adversary, we can also use an analogue of Lemma~\ref{lem:new-additive-classical} to choose the adversary matrix in the special case of classical problems. This implies that our method is an extension of \v{S}palek's original multiplicative adversary method~\cite{Spa08}.

\section{Comparison of the adversary methods}
\begin{definition}
We define the \emph{additive adversary bound} and the \emph{hybrid adversary bound} respectively as
\begin{align*}
\NADV_\varepsilon(\PP)=\max_{\aGamma} \frac{1-C(\varepsilon)}{\max_{x}\norm{\aGamma-\aGamma_x}} \qquad \text{and} \qquad \ADV_\varepsilon(\PP)=\max_{\aGamma,\alambda<1} \frac{\tilde{K}(\aGamma,\alambda,\varepsilon)}{\max_{x}\norm{\aGamma-\aGamma_x}}
\end{align*}
where, for $\NADV$, the maximum is taken over additive adversary matrices $\aGamma$ such that $\Tr\left[\aGamma(\target \circ M)\right]=0$ for any junk matrix $M$, while for $\ADV$ it is taken over all additive adversary matrices.
%Similarly, we define the hybrid adversary bound as
%\begin{align*}
%\ADV_\varepsilon(\PP)=\max_{\aGamma,\alambda} \frac{\tilde{K}(\aGamma,\alambda,\varepsilon)}{\max_{x}\norm{\aGamma-\aGamma_x}},
%\end{align*}
%where the maximum is taken over all additive adversary matrices. 
Finally, we define the \emph{multiplicative adversary bound} as
\begin{align*}
 \MADV_\varepsilon(\PP)=\sup_{\lambda>1} \MADV_\varepsilon^{(\lambda)}(\PP) 
 \ \ \text{where} \ \
 \MADV_\varepsilon^{(\lambda)}(\PP)=\sup_{\Gamma} \tfrac{\log K(\Gamma,\lambda,\varepsilon)}{\log
\max \left\{\norm{\Gamma_{x}^{1/2}\Gamma^{-1/2}}^2
,\norm{\Gamma^{1/2}\Gamma_{x}^{-1/2}}^2:\forall x\in\II\right\}},
\end{align*}
and the supremum is taken over all multiplicative adversary matrices $\Gamma$.
\end{definition}

In this section, we show that the three methods are progressively stronger (the two inequalities are proved independently in the next two sections).
\begin{thm}\label{thm:comparison-methods}
$\MADV_{\varepsilon}(\PP) \geq \ADV_\varepsilon(\PP)\geq \NADV_\varepsilon(\PP)/60$.
\end{thm}

\subsection{Additive versus hybrid}
We show that the hybrid adversary method is always at least as strong as the original additive one (up to a constant factor).
\begin{lem}\label{lem:new-vs-original-AAM}
$\ADV_\varepsilon(\PP)\geq \NADV_\varepsilon(\PP)/60$.
\end{lem}
The proof of this lemma relies on the following.
\begin{lem}\label{lem:new-vs-original-AAM}
 Let $\aGamma$ be an additive adversary method such that $\Tr\left[\aGamma(\target \circ M)\right]=0$ for any junk matrix $M$. Then, for any $\alambda,\varepsilon$ such that $\frac{\varepsilon}{1-\varepsilon}\leq\alambda\leq 1$, we have
\begin{align*}
 \tilde{K}(\aGamma,\alambda,\varepsilon)>(1-\alambda)\left(\sqrt{1-\varepsilon}-\frac{1}{\sqrt{1+\alambda}}\right)^2.
\end{align*}
\end{lem}
\begin{proof}
 Let $V_\bad$ be the direct sum of eigenspaces of $\aGamma$ with eigenvalue strictly larger than $\alambda$. From the definition of $\tilde{K}(\aGamma,\alambda,\varepsilon)$, it suffices to show that
$\Tr\left[\Pi_\bad(\target \circ M)\right]< 1/(1+\alambda)$ for any junk matrix $M$.
Let 
$p_\bad=\Tr\left[\Pi_\bad(\target \circ M)\right]=\norm{\Pi_\bad\ket{\psitarget_M}}^2 $,
where $\ket{\psitarget_M}$ is defined as above. Let us also define the states  $\ket{\psi_\bad}=\Pi_\bad\ket{\psitarget_M}/\sqrt{p_\bad}$ and $\ket{\psi_\good}=\Pi_\good\ket{\psitarget_M}/\sqrt{1-p_\bad}$, so that $\ket{\psitarget_M}=\sqrt{p_\bad}\ket{\psi_\bad}+\sqrt{1-p_\bad}\ket{\psi_\good}$. From the properties of the additive adversary matrix $\aGamma$, we have
\begin{align*}
 0&=\Tr\left[\aGamma(\target\circ M)\right]
=\Tr\left[\aGamma\proj{\psitarget_M}\right] =p_\bad\Tr\left[\aGamma\proj{\psi_\bad}\right]+(1-p_\bad)\Tr\left[\aGamma\proj{\psi_\good}\right]\\
&> p_\bad \alambda+(1-p_\bad)(-1)= (\alambda+1)p_\bad-1.
\end{align*}
This implies that $p_\bad< 1/(1+\alambda)$.
\end{proof}

%We are now ready to prove Lemma~\ref{lem:new-vs-original-AAM}.
\begin{proof}[Proof of Lemma~\ref{lem:new-vs-original-AAM}]
 This is immediate for $\varepsilon\geq 1/5$ as in this case, we have $\NADV_\varepsilon(\PP)=0$. Therefore, it suffices to show that for any additive adversary matrix $\aGamma$ and any $\varepsilon<1/5$, we have $\max_{\alambda}\tilde{K}(\aGamma,\alambda,\varepsilon)\geq (1-\varepsilon-2\sqrt{\varepsilon(1-\varepsilon)})/60$. Let
\begin{align*}
 \alambda=\left(\frac{4}{1-\varepsilon}\right)^{1/3}-1,
\end{align*}
and note that $\frac{\varepsilon}{1-\varepsilon}\leq\alambda\leq 1$ when $0\leq\varepsilon\leq 1/2$.
% \comment{This value of $\alambda$ maximize the bound in Lemma~\ref{lem:new-vs-original-AAM}}.
By Lemma~\ref{lem:new-vs-original-AAM}, we then have
\begin{align*}
 \max_{\alambda}\tilde{K}(\aGamma,\alambda,\varepsilon)
\geq 1-2\varepsilon-3(2-2\varepsilon)^{2/3}+3(2-2\varepsilon)^{1/3}
\geq (1-\varepsilon-2\sqrt{\varepsilon(1-\varepsilon)})/60,
\end{align*}
for any $0\leq\varepsilon\leq 1/2$.
\end{proof}
% \comment{The constant $60$ comes from the values at $\varepsilon=0$, where $1-\varepsilon-2\sqrt{\varepsilon(1-\varepsilon)}=1$ but $1-2\varepsilon-3(2-2\varepsilon)^{2/3}+3(2-2\varepsilon)^{1/3}\approx 0.01756>1/60$.}

\subsection{Hybrid versus multiplicative}
We now show that the multiplicative adversary method is as always at least as strong as the hybrid one.

\begin{lem}\label{lem:new-vs-MAM}
$\lim_{\lambda\to 1}\MADV_\varepsilon^{(\lambda)}(\PP)\geq \ADV_\varepsilon(\PP)$.
\end{lem}
%Note that together with Lemma~\ref{lem:new-vs-original-AAM}, this immediately implies Theorem~\ref{thm:comparison-methods}.
\begin{proof}
 Let $\aGamma$ be the additive adversary matrix achieving $\ADV_\varepsilon(\PP)$.
Therefore, we have
\begin{align*}
 \ADV_\varepsilon(\PP)=\frac{\tilde{K}(\aGamma,\alambda,\varepsilon)}{\max_{x}\norm{\aGamma-\aGamma_x}}.
\end{align*}
Let $\Gamma(\gamma)=\I+\gamma(\I-\aGamma)$. Since $\aGamma\ket{\delta}=\norm{\aGamma}=1$, we see that for any $\gamma>0$, $\Gamma(\gamma)$ is definite positive with $\Gamma(\gamma)\succeq \I$ and $\Gamma(\gamma)\ket{\delta}=1$,
therefore it is a valid multiplicative adversary matrix. Moreover, $\Gamma$ has eigenvalue at least $\lambda=1+\gamma(1-\alambda)$ over $V_\good$. Therefore, $K(\Gamma(\gamma),\lambda(\gamma),\varepsilon)=1+\gamma \tilde{K}(\aGamma,\alambda,\varepsilon)$ and, by definition of the multiplicative adversary bound,
\begin{align*}
 \MADV_\varepsilon(\PP)\geq\sup_{\gamma>0} \frac{\ln \left[1+\gamma \tilde{K}(\aGamma,\alambda,\varepsilon))\right]}{\ln
\max \left\{\norm{\Gamma_{x}^{1/2}(\gamma)\Gamma^{-1/2}(\gamma)}^2
,\norm{\Gamma^{1/2}(\gamma)\Gamma_{x}^{-1/2}(\gamma)}^2:\forall x\in\II\right\}
}.
\end{align*}
We show that in the limit $\gamma\xrightarrow[>]{} 0$, the argument of the supremum is just $\ADV_\varepsilon(\PP)$, which implies the lemma. For the numerator, we immediately have
\begin{align*}
 \ln \left[1+\gamma \tilde{K}(\aGamma,\lambda,\varepsilon))\right]= \gamma \tilde{K}(\aGamma,\lambda,\varepsilon)+\OO(\gamma^2).
\end{align*}
Also, since $\Gamma_x(\gamma)=\I+\gamma(\I-\aGamma_x)$, we have
% \begin{align*}
%  \Gamma_{x}^{1/2}(\gamma)&=\I+\frac{\gamma}{2}(\I-\aGamma_x)+O(\gamma^2)\\
% \Gamma^{-1/2}(\gamma)&=\I-\frac{\gamma}{2}(\I-\aGamma)+O(\gamma^2)
% \end{align*}
% and
\begin{align*}
\norm{\Gamma_{x}^{1/2}(\gamma)\Gamma^{-1/2}(\gamma)}^2
&=\norm{\I+\frac{\gamma}{2}(\aGamma-\aGamma_x)}^2+\OO(\gamma^2),\\
\norm{\Gamma^{1/2}(\gamma)\Gamma_{x}^{-1/2}(\gamma)}^2
&=\norm{\I-\frac{\gamma}{2}(\aGamma-\aGamma_x)}^2+\OO(\gamma^2).
\end{align*}
%Together, these imply
%\begin{align*}
%\max \left\{\norm{\Gamma_{x}^{1/2}(\gamma)\Gamma^{-1/2}(\gamma)}^2
%,\norm{\Gamma^{1/2}(\gamma)\Gamma_{x}^{-1/2}(\gamma)}^2\right\}
%&=\left(1+\frac{\gamma}{2}\norm{\aGamma-\aGamma_x}\right)^2+\OO(\gamma^2)\\
%&=1+\gamma\norm{\aGamma-\aGamma_x}+\OO(\gamma^2).
%\end{align*}
Therefore, we have for the denominator
\begin{align*}
 L(\gamma,x) \stackrel{\rm def}{=} \ln
\max \left\{\norm{\Gamma_{x}^{1/2}(\gamma)\Gamma^{-1/2}(\gamma)}^2
,\norm{\Gamma^{1/2}(\gamma)\Gamma_{x}^{-1/2}(\gamma)}^2\right\}
&=\gamma\norm{\aGamma-\aGamma_x}+O(\gamma^2).
\end{align*}
Since $\lim_{\gamma \rightarrow 0}L(\gamma,x)$ exists for all $x$ and there are only a finite number of possible $x$, we can swap lim and max, which finally implies that:
\begin{align*}
 \lim_{\gamma\to 0} \frac{\ln \left[1+\gamma \tilde{K}(\aGamma,\lambda,\varepsilon))\right]}{\ln
\max \left\{\norm{\Gamma_{x}^{1/2}(\gamma)\Gamma^{-1/2}(\gamma)}^2
,\norm{\Gamma^{1/2}(\gamma)\Gamma_{x}^{-1/2}(\gamma)}^2:\forall x\in\II\right\}
}
&= \ADV_\varepsilon(\PP).
\end{align*}
% \comment{I think I fixed the swap of lim and max.}
\end{proof}

%%%%%%%%%%%%%%%%%%%%
\section{Strong direct product theorem}

In this section we extend \v{S}palek's strong direct product theorem~\cite{Spa08} to quantum state generation problems.
We prove that for any problem which accepts a multiplicative adversary bound $\MADV_\varepsilon^{(\lambda)}(\PP)$, if one wants to solve $\PP^{(k)}$, i.e., $k$ independent instances of $\PP$, using less than $k/10$ times the number of queries necessary to solve one instance with error $\varepsilon$, then the success probability for $\PP^{(k)}$ is exponentially small in $k$. Let us note that a similar theorem was recently proved for the polynomial method~\cite{She10}.

\begin{thm}[Strong direct product]
\label{thm:SDPT}
For any problem $\PP$ and $\lambda>1$,
%  such that 
% \begin{itemize*}
% \item for all normalized Gram matrices $M$,  $\Tr[\Pi_\bad(\target\circ M)] \leq \eta \leq 1/2$ 
% \item $\MADV_\varepsilon(\PP)$ is attained for an adversary matrix $\Gamma$ and a threshold $\lambda$
% \end{itemize*}
% then
there exist a constant $0<c<1$ and an integer $k_0 >0$ such that, for any
$k>k_0$, we have $\MADV_{1-c^k}^{(\lambda)}(\PP^{(k)}) \geq \frac{k}{10}\cdot \MADV_{\varepsilon}^{(\lambda)}(\PP)$.
% \begin{align*}
% 	\forall k>k_0,\ \MADV_{\varepsilon'}^{(\lambda)}(\PP^{(k)}) \geq \frac{k}{10}\cdot \MADV_{\varepsilon}^{(\lambda)}(\PP)
% \end{align*}
% where $\sqrt{1-\varepsilon'} = c^{k} + \eta^{k/5}$.
% where $\varepsilon' = 1-c^{k}$.
\end{thm}
\begin{proof}
This proof closely follows the footsteps of the one by \v{S}palek in~\cite[Sec.~5]{Spa08}, which dealt with the special case of computing functions. Let us assume that the multiplicative adversary bound for $\PP$ with threshold $\lambda$ is obtained by the adversary matrix $\Gamma$. For $\PP^{(k)}$, we construct an adversary matrix $\Gamma' = \Gamma^{\otimes k}$ and set the threshold at value $\lambda' = \lambda^{\frac{k}{10}}$.

First of all we observe that $\max_{x\in\Sigma_I,i\in[k]} \norm{\Gamma_{x,i}^{\prime 1/2}\Gamma^{\prime -1/2}}=\max_{x\in\Sigma_I} \norm{\Gamma_x^{1/2}\Gamma^{-1/2}}$ where $i$ is the index of the queried oracle and $\Gamma'_{x,i} = \Gamma' \circ (\I^{i-1} \otimes D_x \otimes \I^{k-i})$. The proof follows by noting that for $x\in\Sigma_{I}$ and for all $i\in[k]$ we have
\begin{align*}
	\Gamma_{x,i}^{\prime 1/2}\Gamma^{\prime -1/2} &= \left({\Gamma^{1/2}}^{\otimes i-1} \otimes \Gamma_x^{1/2} \otimes {\Gamma^{1/2}}^{\otimes k-i}\right)\left({\Gamma^{-1/2}}^{\otimes i-1} \otimes \Gamma^{-1/2} \otimes {\Gamma^{-1/2}}^{\otimes k-i}\right) \\
		&= \I^{\otimes i-1} \otimes \Gamma_x^{1/2}\Gamma^{-1/2} \otimes \I^{\otimes k-i}.
\end{align*}
We can do the same calculation for the uncomputing oracle.

Let us now find an upper bound to $\max_M \Tr[\Pi'_\bad(\target\circ M)]$. The ``bad'' subspace $V'_\bad$ for the problem $\PP^{(k)}$ is defined by the direct sum of eigenspaces of $\Gamma^{\otimes k}$ with eigenvalue at most $\lambda'=\lambda^{k/10}$.
While, we do not have in general $V'_\bad \subset V_\bad^{\otimes k}$ nor  $V_\bad^{\otimes k} \subset V'_\bad$, we know that $V'_\bad$ is a subspace of the direct sum of spaces $\bigotimes_{i=1}^k V_{v_i}$ where $v\in\{\good,\bad\}^k$ and the number of good subspaces $|v|$ is at most $\frac{k}{10}$.
Indeed, any other eigenspace of $\Gamma'$ has eigenvalue at least $1^{9k/10}\lambda^{k/10}=\lambda'$ since the eigenvalues of $\Gamma$ are greater than $1$, and those associated to good subspaces are greater than $\lambda>1$. Therefore, the projector $\Pi'_\bad$ on the bad subspace is such that $\Pi'_\bad = \Pi'_\bad\cdot\left(\bigoplus_v \bigotimes_i \Pi_{v_i}\right)$. Let us consider a junk matrix $M'$ for $\PP^{(k)}$. Such a matrix can be written as $M'=\sum_j m_j \bigotimes_{i=1}^k M_{i,j}$ where $\sum_j m_j = 1$, and each $M_{i,j}$ is a junk matrix for $\PP$.
\begin{align*}
	\Tr[\Pi'_\bad (\target^{\otimes k}\circ M')] 
	&\leq \sum_{v,j} m_j \Tr\left[\bigotimes_{i=1}^k \Pi_{v_i}(\target\circ M_{ij})\right] \\
	&= \sum_{v,j} m_j \prod_i \Tr[\Pi_{v_i}(\target\circ M_{ij})] \\
	&\leq \sum_{v,j} m_j \eta^{9k/10} \\
	&\leq \eta^{9k/10} \sum_{v: |v|<k/10} 1 \\
	&\leq \eta^{2k/5} \qquad \text{ for } \eta \leq 1/2 \text{ and } k\geq 361
\end{align*}
We conclude that we can take $\eta' = \eta^{2k/5}$.
Let us also define the constants $\zeta = (\sqrt{1-\varepsilon}-\sqrt{\eta})^2$ and $\zeta_0 = \left(\frac{K(\Gamma,\lambda,\varepsilon)}{\lambda}\right)^{1/10} = \left(\frac{1+(\lambda-1)\zeta}{\lambda}\right)^{1/10} <1$ since $\lambda > 1$. There exists $k_0>361$ and $0<c<1$ such that for all $k > k_0,\ \zeta_0^{k/2} + \eta^{k/5} \leq c^{k/2}$. For such $k$'s, we choose $\varepsilon' = 1-c^k$. With these choices, we have
\begin{align*}
	K(\Gamma',\lambda',\epsilon')  \geq 1 + (\lambda' -1)\zeta_0^{k} = 1 + (1-\lambda^{-k/10})K(\Gamma,\lambda,\varepsilon)^{k/10}  \geq K(\Gamma,\lambda,\varepsilon)^{k/10},
\end{align*}
where we used the fact that $K(\Gamma,\lambda,\varepsilon)<\lambda$. Combining everything, we then have
\begin{align*}
	\frac{k}{10} \MADV_\varepsilon(\PP) &= \frac{\ln K(\Gamma,\lambda,\varepsilon)^{k/10}}{\ln
\max \left\{\norm{\Gamma_{x}^{1/2}\Gamma^{-1/2}}^2
,\norm{\Gamma^{1/2}\Gamma_{x}^{-1/2}}^2:\forall x\in\II\right\}} \\
 &\leq \frac{\ln K(\Gamma',\lambda',\varepsilon')}{\ln
\max \left\{\norm{\Gamma_{x}^{\prime 1/2}\Gamma^{\prime -1/2}}^2
,\norm{\Gamma^{\prime 1/2}\Gamma_{x}^{\prime-1/2}}^2:\forall x\in\II\right\}} \leq \MADV_{\epsilon'}(\PP^{(k)}).
\end{align*}
\end{proof}

Let us note that while we have proved that the multiplicative adversary method is stronger than the additive one, we cannot directly conclude that this strong direct product theorem also applies to the additive bound. This is because we can only prove that the multiplicative adversary method becomes stronger in the limit of $\lambda$ going to 1, while in the same limit the constant $c$ in the theorem also goes to 1. Therefore, this only implies a direct sum theorem for the additive adversary bound.

% To conclude that it also satisfies a direct product theorem, one approach would be to prove that whenever the multiplicative adversary method can prove a lower bound in the limit $\lambda\to 1$, there exists some fixed $\lambda>1$ which leads to the same bound. The most important consequence would be for the quantum query complexity of functions, which would itself satisfy a strong direct product theorem for any function, since the additive adversary method is known to be tight in that case~\cite{Rei09,LMRS10}.
% 
% At first sight, this theorem looks weaker~--- though extended to quantum state generation~--- than the one in \cite{Spa08} since we add the condition on the multiplicative adversary bound being attained by a threshold $\lambda$ and a matrix $\Gamma$. In fact, in that paper, the definition of $K_{\rm Spa08}(\lambda,\Gamma,\epsilon)$ is slightly weaker than the one we use here: $K_{\rm Spa08}(\lambda,\Gamma,\epsilon) = \lambda(\sqrt{1-\varepsilon}-\sqrt{\eta})^2 < K(\lambda,\Gamma,\epsilon)$. With that definition, we can relax this condition. Moreover $k_0$ can be made as small as desired by adjusting some constants. Unfortunately, using this definition, we are unable to prove that the multiplicative adversary method is stronger than the additive one which would allow us to conclude that any classical problem satisfies the strong direct product theorem.

\section{Representation theory}\label{sec:representation-theory}
\subsection{Symmetrization of the circuit}
In this section we will study how the symmetries of the problem can help choosing the adversary matrix and in turn obtain the lower bounds. Recall that the oracle computes a function $f\in\Fset$ from $\Sigma_{I}$ to $\Sigma_{O}$, where 
the input alphabet has size $N=|\Sigma_{I}|$ and the output alphabet has size $M=|\Sigma_{O}|$. Let us consider permutations $(\pi,\tau)\in S_{N}\times S_{M}$ acting on $f\in\Fset$ as
\begin{align*}
	f_{\pi,\tau} &= \tau \circ f \circ \pi,
\end{align*}
that is, $f_{\pi,\tau}:\Sigma_{I}\mapsto\Sigma_{O}:x\mapsto \tau(f(\pi(x)))$.

\begin{definition}[Automorphism group of $\PP$]\label{def:automorphism}
We call a group $G\subseteq S_{N}\times S_{M}$ an \emph{automorphism group} of a problem $\PP$ if
\begin{itemize*}
 \item For any $(\pi,\tau)\in G$ and $f\in\Fset$, we have $f_{\pi,\tau}\in\Fset$.
 \item For any $(\pi,\tau)\in G$, there exists a unitary $V_{\pi,\tau}$ such that $V_{\pi,\tau}\ket{\PP(f)}=\ket{\PP(f_{\pi,\tau})}$ for all $f\in\Fset$.
\end{itemize*}
\end{definition}

Note that from an oracle for $f$, it is easy to simulate an oracle for $f_{\pi,\tau}$ by prefixing and appending the necessary permutations on the input and output registers. Consider for example a computing oracle call. Then, $O_{f_{\pi,\tau}}$ acts on $\ket{x}\ket{0}$ just as $(\pi^{-1}\otimes\tau)O_{f}(\pi\otimes\I)$.

Therefore, if $(\pi,\tau)$ is an element of an automorphism $G$ of $\PP$, we can solve the problem with oracle $f$ in the following indirect way:
\begin{enumerate*}
 \item Solve the problem for $f_{\pi,\tau}$, which will prepare a state close to $\ket{\PP(f_{\pi,\tau})}$.
\item Apply $V_{\pi,\tau}^\dagger$ to map this state to a state close to $\ket{\PP(f)}$.
\end{enumerate*}
Since we want the algorithm to work just as well for any possible $f$, we can use this property to symmetrize the circuit. The idea is to solve the algorithm for $f$ by solving it for $f_{\pi,\tau}$ for all possible $(\pi,\tau)\in G$ simultaneously in superposition. Just as we considered $\ket f$ as an additional input to the circuit, we can also use the same mathematical trick and consider $\ket{\pi,\tau}$ as another input. We then run the algorithm on the superposition $\tfrac{1}{\sqrt{|G|}}\sum_{(\pi,\tau)\in G}\ket{\pi,\tau}$. Note that we can assume without loss of generality that the best algorithm for $\PP$ is symmetrized. Indeed, for any algorithm for $\PP$ with success probability $p$ and query complexity $T$, the symmetrized version will have the same query complexity and a success probability at least $p$. For the same reason, we can also assume that the optimal adversary matrix satisfies a similar symmetry, in the following sense:
%  From now on, we will assume that the algorithm has been symmetrized in such a way. Let us study the symmetries of the reduced state $\rho^t$ for such a symmetrized algorithm.
% 
%  The initial state is then:
% \begin{align*}
% 	\ket{\Psi^{0}} = \ket{0}_{\AA} \left( \inv{\sqrt{|G|}}\sum_{\substack{(\pi,\tau)\in G}} \ket{\pi,\tau}_{\PP} \right) \ket{\delta}_{\FF}.
% \end{align*}
% 
% We show that the function register remains in a state that is invariant under the action of $G$ at any point during the algorithm.
% 
%We note $\ket{\Psi^{t}} = \inv{\sqrt{|\Fset|}} \sum_{f\in\Fset} \ket{\psi_{f}^{t}}_{\AA\PP} \ket{f}_{\FF}$ the state juste after the $t$-th oracle call and
%\begin{align*}
%\rho^t = \Tr_{\AA\PP}\proj{\Psi^{t}} = \inv{|\Fset|}\sum_{f,f'\in\Fset}\braket{\psi_{f}^{t}}{\psi_{f'}^{t}} \ket{f}\!\bra{f'}.
%\end{align*}
% 
%\subsection{Preliminary study of $\rho^{t}$ and $S_{N}\times S_{M}$}
% 
%We have:
%\begin{align*}
%	\rho^0 &= \inv{|\Fset|} \sum_{f,f'\in\Fset} \ket{f}\!\bra{f'} = \proj\delta, \\
%	\target &=  \inv{|\Fset| N} \sum_{f,f'\in\Fset} |\mathrm{Im}(f) \cap \mathrm{Im}(f')| \ket{f}\!\bra{f'}.
%\end{align*}
% 
%The advantage of studying $\rho$ instead of the ``algorithm'' states is that $\rho$ varies only when there is an oracle call: the actions of the unitaries on the $\AA$ part are traced out. 
% 
\begin{lem}\label{lem:symmetry-rho-gamma}
For all $(\pi,\tau)\in G$, let $U_{\pi,\tau}$ be the unitary that maps $\ket f$ onto $\ket{f_{\pi,\tau}}$. Then, we can assume without loss of generality that
% the state $\rho^t$ (for any $t\in[0,T]$) and
the optimal adversary matrix $\Gamma$ satisfies
$
% U_{\pi,\tau}\rho^t U_{\pi,\tau}^{\dagger} &= \rho^{t}\label{eqn:UPiTau-rhot},\\
U_{\pi,\tau}\Gamma U_{\pi,\tau}^{\dagger} = \Gamma
$
 for any $(\pi, \tau)\in G$.
\end{lem}
\noindent\begin{proof}
% For $\rho^t$, this follows directly from the symmetrization of the algorithm.
Let $\Gamma$ be an adversary matrix that does not satisfy this property, and let us consider its symmetrized version
$\bar{\Gamma}=\frac{1}{|G|}\sum_{(\pi,\tau)\in G}U_{\pi,\tau}\Gamma U_{\pi,\tau}^{\dagger}$.

We first show that this matrix is still a valid adversary matrix.
Since $U_{\pi,\tau}\ket{\delta}=\ket{\delta}$ for any $(\pi,\tau)\in G$, we immediately have $\bar{\Gamma}\ket{\delta}=\ket{\delta}$ if $\Gamma\ket{\delta}=\ket{\delta}$. By definition of the automorphism group, we have for any $f,g\in\Fset$ and $(\pi,\tau)\in G$
\begin{align*}
 \braket{\PP(f_{\pi,\tau})}{\PP(g_{\pi,\tau})}
=\bra{\PP(f)}V_{\pi,\tau}^\dagger V_{\pi,\tau}\ket{\PP(g)}=\braket{\PP(f)}{\PP(g)}.
\end{align*}
Therefore, for any junk matrix $M$, we have
\begin{align*}
 \frac{1}{|G|}\sum_{(\pi,\tau)\in G}U_{\pi,\tau} \left(\target\circ M\right)U_{\pi,\tau}^{\dagger}
&=\frac{1}{|G|}\sum_{(\pi,\tau)\in G}\frac{1}{|\Fset|}\sum_{f,g} \braket{\PP(f)}{\PP(g)} M_{fg}\ketbra{g_{\pi,\tau}}{f_{\pi,\tau}}\\
&=\frac{1}{|G|}\sum_{(\pi,\tau)\in G}\frac{1}{|\Fset|}\sum_{f,g} \braket{\PP(f_{\pi,\tau})}{\PP(g_{\pi,\tau})} M_{f_{\pi,\tau}g_{\pi,\tau}}\ketbra{g}{f}\\
&=\frac{1}{|\Fset|}\sum_{f,g} \braket{\PP(f}{\PP(g)}\frac{1}{|G|}\sum_{(\pi,\tau)\in G} M_{f_{\pi,\tau}g_{\pi,\tau}}\ketbra{g}{f}\\
&=\target\circ \bar{M},
\end{align*}
where $\bar{M}$ is the symmetrized version of $M$. Therefore, if $\Gamma$ satisfies $\Tr\left[\Gamma(\target \circ M)\right]=0$ for any junk matrix $M$, we have for $\bar{\Gamma}$,
\begin{align*}
\Tr\left[\bar{\Gamma}(\target \circ M)\right]
=\frac{1}{|G|}\sum_{(\pi,\tau)\in G}\Tr\left[U_{\pi,\tau}\Gamma U_{\pi,\tau}^{\dagger}(\target \circ M)\right]
=\Tr\left[\Gamma(\target \circ \bar{M})\right]=0.
\end{align*}
Similarly, if $\Tr\left[\Pi_\bad(\target \circ M)\right]\leq\eta$ for any junk matrix $M$, where $\Pi_\bad$ is the projector on the bad subspace of $\Gamma$, then
\begin{align*}
 \Tr\left[\bar{\Pi}_\bad(\target \circ M)\right]
=\Tr\left[\Pi_\bad(\target \circ \bar{M})\right]\leq\eta,
\end{align*}
where $\bar{\Pi}_\bad$ is the projector on the bad subspace of $\bar{\Gamma}$.

Let us now show that substituting $\Gamma$ by $\bar{\Gamma}$ can only make the adversary bound stronger. It suffices to show that $\max_x\norm{\bar{\Gamma}-\bar{\Gamma}_x}\leq\max_x\norm{\Gamma-\Gamma_x}$, where $\bar{\Gamma}_x=\bar{\Gamma}\circ D_x$. Recall from Fact~\ref{fact:CP-map} that
$
 \Gamma_x=\sum_{y} \Pi_y^x \Gamma \Pi_y^x
$,
and similarly for $\bar{\Gamma}_x$. By definition of $\Pi_y^x$, we have
$
 U_{\pi,\tau}\Pi_y^x U_{\pi,\tau}^{\dagger}=\Pi^{\pi^{-1}(x)}_{\tau(y)}
$
and in turn
\begin{align*}
 \bar{\Gamma}_x
&=\frac{1}{|G|}\sum_y\sum_{(\pi,\tau)\in G}\Pi_y^xU_{\pi,\tau}\Gamma U_{\pi,\tau}^{\dagger}\Pi_y^x 
 =\frac{1}{|G|}\sum_y\sum_{(\pi,\tau)\in G}U_{\pi,\tau}\Pi^{\pi(x)}_{\tau^{-1}(y)}\Gamma \Pi^{\pi(x)}_{\tau^{-1}(y)}U_{\pi,\tau}^{\dagger} \\
 &=\frac{1}{|G|}\sum_{(\pi,\tau)\in G}U_{\pi,\tau}\Gamma_{\pi(x)}U_{\pi,\tau}^{\dagger}
\end{align*}

Finally, we have
\begin{align*}
 \norm{\bar{\Gamma}-\bar{\Gamma}_x}
&=\frac{1}{|G|}\norm{\sum_{(\pi,\tau)\in G}U_{\pi,\tau}\left[\Gamma-\Gamma_{\pi(x)}\right] U_{\pi,\tau}^{\dagger}} 
 \leq \frac{1}{|G|}\sum_{(\pi,\tau)\in G}\norm{\Gamma-\Gamma_{\pi(x)}} 
 \leq \max_x\norm{\Gamma-\Gamma_x},
\end{align*}
where we have used the triangle inequality.
\end{proof}

Note that the mapping $\rep:(\pi,\tau)\mapsto U_{\pi,\tau}$ defines a representation of the automorphism group $G$ and that
% eqs.~(\ref{eqn:UPiTau-rhot}-\ref{eqn:UPiTau-gamma}) imply that
Lemma~\ref{lem:symmetry-rho-gamma} implies that
% $\rho^t$ and
$\Gamma$ commutes with $U_{\pi,\tau}$ for any $(\pi,\tau)\in G$. This means that the matrices $U_{\pi,\tau}$ and $\Gamma$ block-diagonalize simultaneously in a common basis, where each block corresponds to a different irrep of $G$ in $\rep$. From now on, we will consider the special case where $\rep$ is multiplicity-free. This happens for different interesting problems, such as $t$-\fsearch{}~\cite{ASW07,Spa08} and \IE{} (see Section~\ref{sec:index-erasure}), as a consequence of the following lemma.
\begin{lem}\label{lem:multiplicity-free}
 If, for any $f,g\in\Fset$, there exists $(\pi,\tau)\in G$ such that $g=f_{\pi,\tau}$ and $g_{\pi,\tau}=f$, then $\rep$ is multiplicity-free.
% In that case, $\Gamma$ can be written as
% \begin{align*}
% \Gamma &= \sum_{k} \gamma_{k} \Pi_k,
% \end{align*}
% where $k$ indexes the irreps of $G$ and $\Pi_k$ is the projector onto the space $S_{k}$ spanned by the irrep $k$.
%, $m_{k}$ is its multiplicity and $d_{k}$ is its dimension.
\end{lem}
\begin{proof}
%We give a rough sketch of the proof, using the formalism of association schemes. An association scheme is a set of (0-1)-matrices $\{A_{0},\dots,A_{q}\}$ satisfying 4 conditions:
%1) $A_{i}$ is symmetric;
%2) $\sum_{i=0}^{q} A_{i} = \J$ (the all-1 matrix);
%3) $A_{0} = \I$;
%4) $\{A_{i}\}_{i=0}^k$ forms an algebra, i.e., $A_{i}A_{j}=\sum_{k}p_{ij}^kA_{k}$.
%As a consequence all the $A_{i}$ diagonalize in the same basis and can be rewritten $A_{i} = \sum_{k=1}^q \lambda_{i,k} \Pi_{S_k}$ where the $S_{k}$ are orthogonal subspaces such that $\bigoplus_{k=1}^q S_k=\mathrm{span}\{A_{i}\}$. We can define such a scheme by defining equivalence classes by: $f,g$ are in the same class than $f',g'$ if there exist $(\pi,\tau)\in G$ such that $g=f_{\pi,\tau}$ and $g'=f'_{\pi,\tau}$. For example $A_{0}=I$ will be the matrix corresponding to the equivalence class induced by the trivial relation $f=g$. We then use the orthogonality relation of the eigenvalues of an association scheme, namely $\sum_{k}\lambda_{i,k}\lambda_{j,k}d_{k} \leq 0$ \cite{Bai04} and the contraposition of being multiplicity-free: there exists $i, k_{1}$ and $k_{2}$ such that $\lambda_{i,k_{1}}=\lambda_{i,k_2}$, which leads to a contradiction and therefore implies that $\rep$ is multiplicity-free. The property for $\Gamma$ then directly follows from Lemma~\ref{lem:symmetry-rho-gamma} and $\rep$ being multiplicity-free, using standard representation theory properties.
Let us consider the set of matrices $\MM=\{ A \in \mathbb{C}^{|F| \times|F| } : \forall (\pi,\tau)\in G,\ U_{\pi,\tau} A U_{\pi,\tau}^\dagger = A \}$. It is easy to see that for any $A,B\in\MM,$ we have $AB\in\MM$, therefore $\MM$ defines an algebra.
Note that $\rep$ is multiplicity-free if and only if $\MM$ is commutative, in which case all matrices in $\MM$ diagonalize in a common basis~\cite[p.~65]{Cam99}.
% , and in particular $\Gamma\in\MM$ can be written as $\Gamma = \sum_{k} \gamma_{k} \Pi_k$. We now show that $\MM$ is commutative. Indeed,
For any matrix $A\in\MM$, we have $A^t = A$ since there exists $(\pi,\tau)\in G$ such $\bra{f}A\ket{g} = \bra{f}U_{\pi,\tau}AU_{\pi,\tau}^\dagger\ket{g}=\bra{g}A\ket{f}$. This immediately implies that for any $A,B\in\MM$, we have $AB=(AB)^t=B^tA^t=BA$, therefore $\MM$ is a commutative algebra. (More precisely, it is a Bose-Mesner algebra associated to an association scheme \cite{Bai04}) 
\end{proof}

\subsection{Symmetry of oracle calls}
Recall that oracle calls are closely related to the Hadamard product with $D_x$.
% , in particular, an oracle call maps the state $\rho^t$ to $\rho^{t+1}=\sum_x\rho^t_x\circ D_x$, where $\rho^t_x$ is the reduced density matrix of $\ket{\Psi^t}$ when we condition on the input register being in state $\ket{x}$.
We show that the invariance of $\Gamma$ under the action of a group $G$ implies the invariance of $\Gamma_x=\Gamma\circ D_x$ under the action of the subgroup $G_x$ of $G$ that leaves $x$ invariant.
\begin{lem}\label{lem:symmetry-gammax}
For any $x\in\Sigma_{I}$ and $y\in\Sigma_{O}$, let us define the following subgroups of $G$
\begin{align*}
 G_{xy}&=\{(\pi,\tau)\in G:\pi(x)=x, \tau(y)=y\},\\
 G_x&=\{(\pi,\tau)\in G:\pi(x)=x\}.
\end{align*}
Then $\Pi_y^x$ satisfies
$
U_{\pi,\tau}\Pi_y^x U_{\pi,\tau}^{\dagger} = \Pi_y^x
$
for any $(\pi,\tau)\in G_{xy}$, and $\Gamma_x$ satisfies
$
U_{\pi,\tau}\Gamma_x U_{\pi,\tau}^{\dagger} = \Gamma_x
$
for any $(\pi,\tau)\in G_x$.
\end{lem}
\begin{proof}
Recall that by definition of $\Pi_y^x$, we have
$
 U_{\pi,\tau}\Pi_y^x U_{\pi,\tau}^{\dagger}=\Pi^{\pi^{-1}(x)}_{\tau(y)}
$
for any $(\pi,\tau)\in G$.
This immediately implies the first part of the lemma for $(\pi,\tau)\in G_{xy}$. Moreover, Fact~\ref{fact:CP-map} and Lemma~\ref{lem:symmetry-rho-gamma} imply that
$
  U_{\pi,\tau}\Gamma_x U_{\pi,\tau}^{\dagger}=\Gamma_{\pi^{-1}(x)}
$
for any $(\pi,\tau)\in G$. This implies the second part of the lemma for $(\pi,\tau)\in G_x$.
\end{proof}
Since $\rep$ is a representation of $G$, it is also a representation of the subgroup $G_x$. However, even if $\rep$ is multiplicity-free with respect to $G$, it is typically not with respect to $G_x$. Indeed, when restricting $G$ to $G_x$, multiplicities can happen due to two different mechanisms. First, an irrep can become reducible, and one of the new smaller irreps can be a copy of another irrep. Secondly, two irreps that are different for $G$ could be the same when we restrict to the elements of $G_x$. Let us identify an irrep of $G_x$ by three indices $(k,l,m)$: the first index identifies the irrep $k$ of $G$ from which it originates, the second index identifies the irrep $l$ of $G_x$, and the last index allows to discriminate betwen different copies of the same irrep of $G_x$. For example, two irreps having the same index $l$ but different indices $k$ are two copies of the same irrep of $G_x$ originating from different irreps of $G$. Also, we denote by $V_{k,l,m}$ the subspace spanned by irrep $(k,l,m)$. These subspaces are such that $\bigoplus_{l,m} V_{k,l,m}=V_k$, where $V_k$ is the subspace spanned by the irrep $k$ of $G$ (we assume that $V_{k,l,m}$ is empty if $(k,l,m)$ does not correspond to a valid irrep).
In the following, it will also be useful to define $W_{l} = \bigoplus_{k,m} V_{k,l,m}$ which is sometimes called the isotypical component corresponding to $l$~\cite{Ser77}.

\begin{lem}\label{lem:block-gammax}
Let $\rep$ be multiplicity-free for $G$.
Then, $\Gamma$ can be written as
$
\Gamma = \sum_{k} \gamma_{k} \Pi_k,
$
where $k$ indexes the irreps of $G$ and $\Pi_k$ is the projector onto the space $V_{k}$ spanned by the irrep $k$.
Also, $\Gamma_x$ block-diagonalizes as
$
\Gamma_x = \sum_{l} \Gamma_x^{l},
$
where $l$ indexes the irreps of $G_x$, and, for each $l$, $\Gamma_x^{l}$ is a matrix on the isotypical component  $W_{l}= \bigoplus_{k,m} V_{k,l,m}$  of $l$. Moreover, $\Gamma_x^{l}$ can be written as
\begin{align*}
 \Gamma_x^{l}=\sum_{k_1,m_1,k_2,m_2} \gamma^{l}_{x;k_1m_1;k_2m_2} \Pi^l_{k_1m_1\leftarrow k_2m_2},
\end{align*}
where $d_l$ is the dimension of irrep $l$, $\Pi^l_{k_1m_1\leftarrow k_2m_2}$ is the \emph{``transporter''} from $V_{k_2,l,m_2}$ to $V_{k_1,l,m_1}$, i.e., the operator that maps any state in $V_{k_2,l,m_2}$ to the corresponding state in  $V_{k_1,l,m_1}$, and
\begin{align*}
 \gamma^{l}_{x;k_1m_1;k_2m_2}=\frac{1}{d_l} \Tr\left[ \Gamma_x \Pi^l_{k_2m_2\leftarrow k_1m_1}\right].
\end{align*}
\end{lem}
\begin{proof}
 This directly follows from Lemmas~\ref{lem:symmetry-rho-gamma}-\ref{lem:symmetry-gammax} using the canonical decomposition of the representation $\rep$~\cite{Ser77}.
\end{proof}

\subsection{Computing the adversary bounds}
Lemma~\ref{lem:block-gammax} tells us how to choose the adversary matrix: it suffices to assign weights $\gamma_{k}$ to each irrep $k$ of $G$, i.e., $\Gamma = \sum_{k} \gamma_{k} \Pi_k$. Moreover, it also implies that computing the associated adversary bounds boils down to bounding for each irrep $l$ of $G_x$ the norm of a small $m_l\times m_l$ matrix, where $m_l$ is the multiplicity of irrep $l$.
\begin{thm}\label{thm:adv-representation}
Let $\rep$ be multiplicity-free for $G$. Then, we have
\begin{align*}
\norm{\aGamma_{x}-\aGamma}   = \max_{l} \norm{\tilde{\Delta}_{x}^l}, \qquad\quad
\norm{\Gamma_{x}^{1/2}\Gamma^{-1/2}}^2   = \max_{l} \norm{\Delta_{x}^l},\qquad\quad
\norm{\Gamma^{1/2}\Gamma_{x}^{-1/2}}^2   = \max_{l} \norm{(\Delta_{x}^l)^{-1}},
\end{align*}
where the maximums are over irreps $l$ of $G_x$. For each irrep $l$, $\tilde{\Delta}_{x}^l$ and $\Delta_x^l$ are $m_l\times m_l$ matrices, where $m_l$ is the multiplicity of $l$ for $G_x$, with elements labeled by the different copies of the irrep and such that
\begin{align*}
(\tilde{\Delta}^{l}_{x})_{k_1m_1,k_2,m_2}
&=  \frac{1}{d_l}\sum_{k,y} \gamma_{k}\Tr\left[\Pi_y^x\Pi_k\Pi_y^x \Pi^l_{k_1m_1\leftarrow k_2m_2}\right]-\gamma_{k_1}\delta_{k_1k_2}\\
(\Delta^{l}_{x})_{k_1m_1,k_2,m_2} &=  \frac{1}{d_l}\sum_{k,y} \frac{\gamma_{k}}{\sqrt{\gamma_{k_{1}}\gamma_{k_{2}}}}\Tr\left[\Pi_y^x\Pi_k\Pi_y^x \Pi^l_{k_1m_1\leftarrow k_2m_2}\right].
\end{align*}
\end{thm}
\begin{proof}
 This follows directly from Lemma~\ref{lem:block-gammax} and the definition of $\Gamma_x$.
\end{proof}

We see that to obtain the adversary bounds, we need to compute the traces of products of four operators. Since $G_{xy}$ is a subgroup of both $G$ and $G_x$, each of these operators can be decomposed into a sum of projectors onto irreps of $G_{xy}$ (or transporters from and to these irreps). To compute these traces, we can use the following lemma, which shows that it is sufficient to compute the traces of products of two projectors onto irreps of $G_{xy}$.
\begin{lem}\label{lem:trace-proj}
	Let $\lambda,\mu,\nu_1,\nu_2$ denote irreps of $G_{xy}$. If any of $\mu,\nu_1$ or $\nu_2$ is not isomorphic to $\lambda$, then
$
		\Tr\left[\Pi_\lambda\Pi_\mu\Pi_\lambda \Pi_{\nu_1\leftarrow \nu_2}\right]
		=0.
$
Otherwise, we have
	\begin{align*}
		\Tr\left[\Pi_\lambda\Pi_\mu\Pi_\lambda \Pi_{\nu_1\leftarrow \nu_2}\right]
		&=\frac{1}{d}\Tr\left[\Pi_\lambda\Pi_\mu\right]\cdot\Tr\left[\Pi_\lambda \Pi_{\nu_1\leftarrow\nu_2}\right], \\
		\abs{\Tr\left[\Pi_\lambda \Pi_{\nu_1\leftarrow \nu_2}\right]}
		&=\sqrt{\Tr\left[\Pi_\lambda \Pi_{\nu_1}\right]\cdot\Tr\left[\Pi_\lambda \Pi_{\nu_2}\right]},
	\end{align*}
where $d$ is the dimension of the representation $\lambda$.
\end{lem}
\begin{proof}
If two irreps are not isomorphic to each other, they belong to different isotypical subspaces of $\rep$, and therefore the product of their projectors (or transporters) is zero. Let us now assume that all the irreps are isomorphic to $\lambda$, and therefore belong to the same isotypical subspace. Then, we can define isomorphic bases $\{\ket{i}\}_{i\in[d]},\{\ket{\psi_i}\}_{i\in[d]}$, $\{\ket{\phi_i^{(1)}}\}_{i\in[d]}$ and $\{\ket{\phi_i^{(2)}}\}_{i\in[d]}$ for the subspaces spanned by irreps $\lambda, \mu,\nu_1$ and $\nu_2$, respectively, such that
\begin{align*}
 \Pi_\lambda&=\sum_{i=1}^d\proj{i}, &
 \Pi_\mu&=\sum_{i=1}^d\proj{\psi_i}, &
 \Pi_{\nu_1\leftarrow \nu_2}&=\sum_{i=1}^d\ket{\phi_i^{(1)}}\bra{\phi_i^{(2)}}.
\end{align*}
Let us also choose a basis $\{\ket{i,j}\}_{(i,j)\in[d]\times[m]}$ for the whole $(d\times m)$-dimensional isotypical subspace, $m$ being the multiplicity of the irreps. Without loss of generality, we may choose this basis such that $\{\ket{i,1}\}_{i\in[d]}=\{\ket{i}\}_{i\in[d]}$ corresponds to $\lambda$ itself, and, for any $j\neq 1$, $\{\ket{i,j}\}_{i\in[d]}$ corresponds to a copy of $\lambda$. Since $\lambda, \mu,\nu_1$ and $\nu_2$ are isomorphic, there exists coefficients $\{\alpha_j\}_{j\in[m]}, \{\beta_j^{(1)}\}_{j\in[m]}$ and $\{\beta_j^{(2)}\}_{j\in[m]}$ such that 
\begin{align*}
 \ket{\psi_i}&=\sum_{j=1}^m\alpha_j\ket{i,j}, &
 \ket{\phi_i^{(1)}}&=\sum_{j=1}^m\beta_j^{(1)}\ket{i,j}, &
 \ket{\phi_i^{(2)}}&=\sum_{j=1}^m\beta_j^{(2)}\ket{i,j}.
\end{align*}
We now have
\begin{align*}
 \Tr\left[\Pi_\lambda\Pi_\mu\Pi_\lambda \Pi_{\nu_1\leftarrow \nu_2}\right]
		&=\sum_{i=1}^d\braket{i}{\psi_i}\braket{\psi_i}{i}\braket{i}{\phi_i^{(1)}}\braket{\phi_i^{(2)}}{i}\\
&=d\cdot\braket{1}{\psi_1}\braket{\psi_1}{1}\braket{1}{\phi_1^{(1)}}\braket{\phi_1^{(2)}}{1}\\
&=\frac{1}{d}\sum_{i=1}^d\braket{i}{\psi_i}\braket{\psi_i}{i}\cdot\sum_{j=1}^d\braket{j}{\phi_j^{(1)}}\braket{\phi_j^{(2)}}{j}\\
&=\frac{1}{d}\Tr\left[\Pi_\lambda\Pi_\mu\right]\cdot\Tr\left[\Pi_\lambda \Pi_{\nu_1\leftarrow\nu_2}\right].
\end{align*}
Similarly, we also have
\begin{align*}
 \Tr\left[\Pi_\mu \Pi_{\nu_1\leftarrow \nu_2}\right]
\cdot\Tr\left[\Pi_\mu \Pi_{\nu_2\leftarrow \nu_1}\right]
&=\sum_{i=1}^d\braket{i}{\phi_i^{(1)}}\braket{\phi_i^{(2)}}{i}\cdot\sum_{j=1}^d\braket{j}{\phi_j^{(2)}}\braket{\phi_j^{(1)}}{j}\\
&=d^2\cdot\braket{1}{\phi_1^{(1)}}\braket{\phi_1^{(2)}}{1}\braket{1}{\phi_1^{(2)}}\braket{\phi_1^{(1)}}{1}\\
&=\sum_{i=1}^d\braket{i}{\phi_i^{(1)}}\braket{\phi_i^{(1)}}{i}\cdot\sum_{j=1}^d\braket{j}{\phi_j^{(2)}}\braket{\phi_j^{(2)}}{j}\\
	&=\Tr\left[\Pi_\mu \Pi_{\nu_1}\right]\cdot\Tr\left[\Pi_\mu \Pi_{\nu_2}\right].
\end{align*}
\end{proof}

\section{Applications}
\subsection{\search}

By considering Grover's \search{} problem~\cite{Gro96}, which we denote $\search_n$, we can show that the inequalities in Theorem~\ref{thm:comparison-methods} are strict.

\begin{thm}\label{thm:search}
For any $0<\varepsilon<1-\frac{1}{n}$, we have
\begin{align*}
 \NADV_\varepsilon(\search_n)&=\Omega\left((1-\varepsilon-2\sqrt{\varepsilon(1-\varepsilon)})\sqrt{n}\right)\\
 \ADV_\varepsilon(\search_n)&=\Omega\left((\sqrt{1-\varepsilon}-1/\sqrt{n})^2\sqrt{n}\right)\\
\MADV_\varepsilon(\search_n)&=\Omega\left((\sqrt{1-\varepsilon}-1/\sqrt{n})\sqrt{n}\right).
\end{align*}
In particular, for $\varepsilon>1/5$, we have $\MADV_\varepsilon(\search_n)>\ADV_\varepsilon(\search_n)> \NADV_\varepsilon(\search_n)$.
\end{thm}
In order to illustrate our method, we will use representation theory to compute the adversary bounds, even though this is not really necessary for such a simple problem. The $\Omega(\sqrt{n})$ lower bound for large success probability is well-known (see e.g~\cite{BBB+97}), and the case of small success probability has been studied in~\cite{Amb05, Spa08} using the multiplicative adversary method. The fact that a non-trivial bound can also be found in this regime using and additive adversary method (our hybrid method) is new to the present work.

\begin{proof}
Let us denote by $f_x$ the oracle that marks element $x$, that is, $f_x(x')=1$ if $x'=x$ and $0$ otherwise. Let us consider the symmetric group $S_n$ acting on $f$ as $f_\pi(x)=f(\pi(x))$. This groups forms an automorphism for $\search_n$, and the associated representation $\rep$ corresponds to the natural representation acting on $[n]$. This representation decomposes into two irreps, the one-dimensional trivial representation on $V_0=\Span\{\ket{\delta}\}$, where $\ket{\delta}=(1/\sqrt{n})\sum_x \ket{f_x}$, and an $(n-1)$-dimensional irrep on $V_1=V_0^\perp$. Following Lemma~\ref{lem:symmetry-rho-gamma}, we set $\Gamma=\Pi_{0}+\gamma\Pi_{1}$.

Let us now fix some input $x\in\Sigma_{I}$ to the oracle (by symmetry, the calculation will be the same for any $x$). When restricting $G$ to $G_x=\{\pi\in G: \pi(x)=x\}$, the second representation splits into two irreps, the first one being a second copy of the trivial representation, now acting on $V_{1,0}=\Span\{\ket{\delta_x}\}$, where $\ket{\delta_x}=(\ket{\delta}-\sqrt{n}\ket{f_x})/\sqrt{n-1}$. Following our convention, we index the three irreps of $G_x$ with labels $(k,l)$ as $(0,0)$, $(1,0)$ and $(1,1)$ (no need for a third index as each irrep of $G_x$ appears only once in a given irrep of $G$). Since we have one irrep with multiplicity two, and one irrep with multiplicity one, the matrix $\Gamma_x$ will block-diagonalize into two blocks: one $2\times2$ block $\Gamma_x^0$ on $V_0\oplus V_{1,0}$, and one $(n-1)\times(n-1)$ block $\Gamma_x^1$ on $V_{1,1}$.

It is easy to check that only the block corresponding to the trivial representation $l=0$ is relevant. Indeed, since the other representation has multiplicity 1, the corresponding block is characterized by a single scalar, and it is straightforward to check that $\tilde{\Delta}_{x}^1=0$ and $\Delta_{x}^1=1$, so that the maximum in Theorem~\ref{thm:adv-representation} will not be achieved by this block.

Let us now consider the other representation, corresponding to a $2\times2$ block. In order to compute matrices $\tilde{\Delta}_{x}^0$, and $\Delta_{x}^0$, we first need to compute $\Pi_{0}\circ D_x$ and $\Pi_{1}\circ D_x$, which is straightforward using Fact~\ref{fact:CP-map}. In the basis $\{\ket{\delta},\ket{\delta_x}\}$, we obtain
\begin{align*}
\Pi_{0}\circ D_x=
 \begin{pmatrix}
  1-2\alpha^2(1-\alpha^2) & \alpha\sqrt{1-\alpha^2}(1-2\alpha^2)\\
\alpha\sqrt{1-\alpha^2}(1-2\alpha^2) & 2\alpha^2(1-\alpha^2)
 \end{pmatrix},
\end{align*}
where $\alpha=1/\sqrt{n}$, and therefore $\Pi_{1}\circ D_x=\I-\Pi_{0}\circ D_x$. For the additive adversary methods, we then obtain from Theorem~\ref{thm:adv-representation}
\begin{align*}
 \tilde{\Delta}_x^0=(1-\gamma)\alpha\sqrt{1-\alpha}
 \begin{pmatrix}
  -2\alpha\sqrt{1-\alpha^2} & 1-2\alpha^2\\
1-2\alpha^2 & 2\alpha\sqrt{1-\alpha^2}
 \end{pmatrix}.
\end{align*}
The matrix has eigenvalues $\pm 1$, so that $\norm{\tilde{\Delta}_x^0}=(1-\gamma)\alpha\sqrt{1-\alpha}$.

For the usual additive adversary method, we need to choose $\gamma$ such that $\Tr(\aGamma (\target\circ M))=0$ for any junk matrix $M$. Here, $\target=\I/n$, therefore this condition reduces to $\Tr(\aGamma)=0$, which is satisfied for $\gamma=-1/(n-1)$. This yields $\norm{\tilde{\Delta}_x^0}=1/\sqrt{n-1}$, and therefore $\NADV_\varepsilon(\search_n)=(1-\varepsilon-2\sqrt{\varepsilon(1-\varepsilon)})\sqrt{n-1}$, which is $\Omega(\sqrt{n})$ for $\varepsilon<1/5$, but negative otherwise.

For the new additive adversary method, we can choose $\alambda=\gamma$, so that $V_\bad=V_0$ and $\eta=1/n$. This implies that as soon as $\varepsilon<1-1/n$, we have a non-trivial bound $\ADV_\varepsilon(\search_n)=(\sqrt{1-\varepsilon}-1/\sqrt{n})^2\sqrt{n-1}$.

For the multiplicative adversary method, we choose $\gamma>1$ and $\alambda=\gamma$, so that $V_\bad=V_0$ and $\eta=1/n$. We then obtain similarly
\begin{align*}
 \Delta_x^0&=
 \begin{pmatrix}
  1+2(\gamma-1)\alpha^2(1-\alpha^2) & -\frac{\gamma-1}{\sqrt{\gamma}}\alpha\sqrt{1-\alpha^2}(1-2\alpha^2)\\
 -\frac{\gamma-1}{\sqrt{\gamma}}\alpha\sqrt{1-\alpha^2}(1-2\alpha^2) & 1-2\frac{\gamma-1}{\gamma}\alpha^2(1-\alpha^2)
 \end{pmatrix}\\
&=
\begin{pmatrix}
  1 & -\frac{\gamma-1}{\sqrt{\gamma}}\alpha\\
 -\frac{\gamma-1}{\sqrt{\gamma}}\alpha & 1
 \end{pmatrix}+O(\alpha^2).
\end{align*}
By Gershgorin circle theorem, the eigenvalues of this matrix lie in the range $[1-\frac{\gamma-1}{\sqrt{\gamma n}},1+\frac{\gamma-1}{\sqrt{\gamma n}}]$, so that
\begin{align*}
 \MADV(\search_n)\geq\frac{\log [1+(\gamma-1)\beta^2]}{\log[1+(\gamma-1)/{\sqrt{\gamma n}}]},
\end{align*}
where $\beta=\sqrt{1-\varepsilon}-1/\sqrt{n}$. In the limit $\gamma\xrightarrow[>]{}1$, we obtain the same bound as for the new additive adversary method. However, for $\gamma= 1+1/\beta^2$, we obtain
\begin{align*}
 \MADV(\search_n)\geq(\log 2)\cdot \frac{\sqrt{\gamma n}}{\gamma-1}=\Omega\left((\sqrt{1-\varepsilon}-1/\sqrt{n})\sqrt{n}\right),
\end{align*}
where we have used the fact that $\log(1+x)\leq x$.
\end{proof}

\subsection{\IE\label{sec:index-erasure}}

Let us now consider the following coherent quantum state generation problem, called \IE~\cite{Shi02}: given an oracle for an injective function $f:[N]\to[M]$, coherently generate the superposition $\ket{\psi_f}=\frac{1}{\sqrt{N}}\sum_{x=1}^N\ket{f(x)}$ over the image of $f$. The name \IE{} comes from the fact that we can easily prepare the superposition $\frac{1}{\sqrt{N}}\sum_{x=1}^N\ket{x}\ket{f(x)}$ using one oracle call, so the problem is to \emph{erase} the index $\ket{x}$.

The previously best known lower bound for \IE{} is $\Omega(\sqrt[5]{N/\log N})$, which follows from a reduction to the \SE{} problem~\cite{Mid04}.
It is also known that this problem may be solved with $O(\sqrt{N})$ oracle calls. Indeed, given $\ket{f(x)}$, one can find the index $\ket{x}$ with $O(\sqrt{N})$ oracle calls using Grover's algorithm for \search~\cite{Gro96}. Therefore, the quantum circuit for this algorithm maps the superposition $\ket{\psi_f}=\frac{1}{\sqrt{N}}\sum_{x=1}^N\ket{f(x)}$ to the the state $\frac{1}{\sqrt{N}}\sum_{x=1}^N\ket{x}\ket{f(x)}$. The algorithm for \IE{} then follows by inverting this circuit.

We now show that this algorithm is optimal by proving a matching lower bound using the hybrid adversary method. 
\begin{thm}
\label{thm:index-erasure}
 $Q_\varepsilon(\IE)=\Theta(\sqrt{N})$.
\end{thm}
\begin{proof}
Let $(\pi,\tau)\in S_N\times S_M$ act on the set $\Fset$ of injective functions from $[N]$ to $[M]$ by mapping $f$ to $f_{\pi,\tau}=\tau\circ f\circ\pi$. Since we can obtain the state $\ket{\psi_f}$ from $\ket{\psi_{f_{\pi,\tau}}}$ by applying the permutation $\tau^{-1}$ on the target register, the whole group $G=S_N\times S_M$ defines an automorphism group for the problem.

\paragraph{Representations.}
Let us study the representation $\rep$ corresponding to the action of $G$ on the set of injective functions $\Fset$. From Lemma~\ref{lem:multiplicity-free}, this representation is multiplicity-free: indeed, for any $f,g\in\Fset$, it is easy to construct a group element $(\pi,\tau)$ that maps both $f$ to $g$ and $g$ to $f$. Therefore, any irrep of $G$ appears in $\rep$ at most once. Let us now show that many irreps do not appear at all. Recall that irreps of $G=S_N\times S_M$ can be represented by pairs of Young diagrams $(\lambda_N,\lambda_M)$, where $\lambda_N$ has $N$ boxes, and $\lambda_M$ has $M$ boxes~\cite{Sag01}. We show that only irreps where the diagram $\lambda_N$ is contained in the diagram $\lambda_M$ can appear. We show this by induction on $M$, starting from $M=N$. For the base case, the set of injective functions $\Fset$ is isomorphic to the set of permutations in $S_N$, and $(\pi,\tau)\in S_N\times S_N$ acts on a permutation $\sigma$ as $\tau\sigma\pi$. Therefore, the only irreps which occur in $\rep$ are those where the two diagrams are the same, that is, $\lambda_N=\lambda_M$. When extending the range of functions in $\Fset$ from $M$ to $M+1$, we induce irreps of $S_N\times S_M$ to irreps of $S_N\times S_{M+1}$ by adding an extra box on the diagram corresponding to $S_M$. Since we start from a case where the two diagrams are the same, we can only obtain pairs of diagrams $(\lambda_N,\lambda_M)$ where $\lambda_N$ is contained inside $\lambda_M$.

\paragraph{Initial and target states.}
The initial state is $\rho_0=\proj{\delta}$, where $\ket{\delta}=\frac{1}{\sqrt{|\Fset|}}\sum_{f\in\Fset}\ket{f}$ is the superposition over all injective functions, which is invariant under any element $(\pi,\tau)\in G$. Therefore, it corresponds to the trivial one-dimensional irrep of $S_N\times S_M$, represented by a pair of diagrams $(\lambda_N,\lambda_M)$ where both diagrams contain only one row of $N$ and $M$ boxes, respectively (see Fig.~\ref{fig:YoungDiagrams}). Let $V_0=\Span\{\ket{\delta}\}$ be the corresponding one-dimensional subspace. We now show that the target state $\target$ is a mixed state over $V_0\oplus V_1$, where $V_1=\Span\{\ket{\phi_y}:y\in[M]\}$ is the $(M-1)$-dimensional subspace spanned by states $\ket{\phi_y}=\sqrt{1-({N}/{M})}\ket{\psi_y}-\sqrt{{N}/{M}}\ket{\bar{\psi}_y}$, $\ket{\psi_y}$ being the uniform superposition over functions $f$ such that $y\in\Im(f)$, and $\ket{\bar{\psi}_y}$ the uniform superposition over functions $f$ such that $y\notin\Im(f)$.
% \begin{align*}
% \ket{\phi_y}&=\sqrt{1-\frac{N}{M}}\ket{\psi_y}-\sqrt{\frac{N}{M}}\ket{\bar{\psi}_y},\\
% \ket{\psi_y}&=\sqrt{\frac{N}{M|\Fset|}}\sum_{f:y\in\Im(f)}\ket{f},\\
% \ket{\bar{\psi}_y}&=\sqrt{\frac{M}{(N-M)|\Fset|}}\sum_{f:y\notin\Im(f)}\ket{f}.
% \end{align*}
This subspace corresponds to the irrep represented by diagrams $(\lambda_N,\lambda_M)$ where $\lambda_N$ contains only one row of $N$ boxes, and $\lambda_M$ has $M-1$ boxes on the first row and one box on the second (see Fig~\ref{fig:YoungDiagrams}). We have for the target state
\begin{align*}
 \target&=\inv{|\Fset|}\sum_{f,f'\in\Fset}\braket{\psi_f}{\psi_{f'}}\ketbra{f'}{f}
=\inv{|\Fset|}\sum_{f,f'\in\Fset}\frac{|\Im(f)\cap\Im(f')|}{N}\ketbra{f'}{f}\\
&=\frac{1}{M}\sum_{y=1}^M\proj{\psi_y}
=\frac{N}{M}\proj{\delta}+\left(1-\frac{N}{M}\right)\frac{1}{M}\sum_{y=1}^M\proj{\phi_y}\\
&=\frac{N}{M}\rho_0+\left(1-\frac{N}{M}\right)\rho_1,
\end{align*}
where $\rho_0$ and $\rho_1$ are the maximally mixed states over $V_0$ and $V_1$, respectively. 

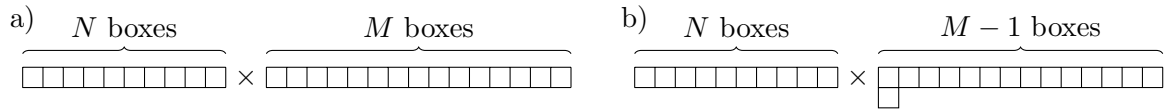
\begin{figure}
\begin{tikzpicture}[scale=.9, decoration=brace]
\node at (0,1) {a)};
\draw[step=.3] (0,0) grid (3,.3); 
\node at (3.3, .15) {$\times$};
\draw[step=.3] (3.6,0) grid (8.1, .3);
\draw[decorate] (0, .5) -- (3, .5); \node at (1.5, .9) {$N$ boxes};
\draw[decorate] (3.6, .5) -- (8.1, .5); \node at (5.85, .9) {$M$ boxes};
\end{tikzpicture}

\vspace*{-1.38cm}

\hspace*{8cm}
\begin{tikzpicture}[scale=.9, decoration=brace]
\node at (0,1) {b)};
\draw[step=.3] (0,0) grid (3,.3);
\node at (3.3, .15) {$\times$};
\draw[step=.3] (3.6,0) grid (7.8, .3);
\draw (3.6,0) rectangle (3.9, -.3);
\draw[decorate] (0, .5) -- (3, .5); \node at (1.5, .9) {$N$ boxes};
\draw[decorate] (3.6, .5) -- (7.8, .5); \node at (5.7, .9) {$M-1$ boxes};
\end{tikzpicture}

\caption{We use $N=10$ and $M=15$. a) Young diagrams corresponding to the one-dimensional space $V_0$. The initial state $\rho^0$ is the projector over $V_0$ ; b) Young diagrams corresponding to the $(M-1)$-dimensional space $V_1$. The target state $\target$ has a large overlap $(1-N/M)$ with the completely mixed state over $V_1$.\label{fig:YoungDiagrams}}
\end{figure}

\paragraph{Adversary matrix.}
Since we start from state $\rho_0$ and we want to reach state $\target$ which has a large weight over $\rho_1$, the strategy for the lower bound is to show that it is hard to transfer weight from $V_0$ to $V_1$. More precisely, we divide all irreps (and by consequence their corresponding subspaces) into two sets: one set of \emph{bad} irreps containing all irreps represented by diagrams $(\lambda_N,\lambda_M)$ where $\lambda_N$ and $\lambda_M$ only differ in their first row, and one set of \emph{good} irreps containing all the other irreps (see Fig.~\ref{fig:YoungDiagramsGoodBad}). By this definition, the irrep corresponding to $V_0$ is bad, while the irrep corresponding to $V_1$ is good. The lower bound is based on the fact that it is hard to transfer weight onto good subspaces (in particular $V_1$) starting from $V_0$. As mentioned in Section~\ref{sec:notations}, from now on, we note the irreps only by their part under the first row; $(\lambda, \lambda')$ then denotes an irrep of $S_N\times S_M$. Therefore, bad irreps are precisely those such that $\lambda=\lambda'$. Recall from Lemma~\ref{lem:symmetry-gammax} that constructing an adversary matrix $\aGamma$ amounts to assigning an eigenvalue to each irrep of $G$. We choose $\aGamma$ such that it has eigenvalue $0$ on good irreps, and eigenvalue $\gamma_{|\lambda|}$ on a bad irrep $(\lambda,\lambda)$, which only depends on $|\lambda|$, i.e.,
\begin{align*}
 \aGamma&=\sum_{\lambda}\gamma_{|\lambda|}\Pi_{(\lambda,\lambda)},
\end{align*}
where $\Pi_{(\lambda,\lambda')}$ is the projector onto the subspace corresponding to the irrep $(\lambda,\lambda')$. We set
\begin{align*}
 \gamma_{|\lambda|}=
\begin{cases}
   1-\frac{|\lambda|}{\sqrt{N}} & \textrm{if } \lambda<\sqrt{N}\\
0 & \textrm{otherwise}.
\end{cases}
\end{align*}
Therefore, we have $\gamma_0=1$ and $0\leq\gamma_{|\lambda|}\leq 1$ for any $\lambda$, and $\aGamma$ is a valid additive adversary matrix. 
 Let $V_\bad$ denote the direct-sum of the bad subspaces. Since $\target$ only has overlap $N/M$ over $V_\bad$, we have $\Tr(\Pi_\bad\target)\leq N/M$. Therefore, we can set the threshold eigenvalue $\alambda=0$ and the base success probability $\eta=N/M$.

\begin{figure}

\begin{tikzpicture}[scale=.85, decoration=brace, step=.3]
\node at (-.8, 1.3) {a)};
\draw (0,0) grid (0.3, 0.3);
\draw (0,.3) grid (0.9,0.6);
\draw (0,.6) grid (1.2,0.9);
\draw (0,.9) grid (2.7,1.2);
\node at (3.3, 1.05) {$\times$};
\draw (3.6,0) grid (3.9, 0.3);
\draw (3.6,.3) grid (4.2,0.6);
\draw (3.6,.6) grid (4.8,0.9);
\draw (3.6,.9) grid (7.5,1.2);

\draw[line width=1.3pt] (0,0) -- (0, 0.9) -- (1.2, 0.9) -- (1.2, 0.6)  -- (0.9, .6) -- (.9, .3) -- (.3, .3) -- (.3, 0) -- cycle;
\draw[line width=1.3pt] (3.6,0) -- (3.6, 0.9) -- (4.8, 0.9) -- (4.8, 0.6)  -- (4.5, .6) -- (4.5, .3) -- (3.9, .3) -- (3.9, 0) -- cycle;

\draw[decorate] (-.15,0) -- (-.15,1.2); \node at(-.7,.6) {$\lambda_N$};
\draw[decorate] (7.65,1.2) -- (7.65,0); \node at (8.2, .6) {$\lambda_M$};

\node (lambda) at (2,.1) {$\lambda$} ;
\draw[->] (lambda) -- (1,.4) ;
\draw[->] (lambda) -- (3.3,.4) ;

\end{tikzpicture}

\vspace*{-1.6cm}

\hspace*{8.3cm}
\begin{tikzpicture}[scale=.85, decoration=brace, step=.3]
\node at (-.4, 1.3) {b)};
\draw (0,0) grid (0.3, 0.3);
\draw (0,.3) grid (0.9,0.6);
\draw (0,.6) grid (.9,0.9);
\draw (0,.9) grid (3,1.2);
\node at (3.3, 1.05) {$\times$};
\draw (3.6,0) grid (4.2, 0.3);
\draw (3.6,.3) grid (4.8,0.6);
\draw (3.6,.6) grid (4.8,0.9);
\draw (3.6,.9) grid (6.9,1.2);

\draw[line width=1.3pt] (0,0) -- (0, 0.9) -- (1.2, 0.9) -- (1.2, 0.6)  -- (0.9, .6) -- (.9, .3) -- (.3, .3) -- (.3, 0) -- cycle;
\draw[line width=1.3pt] (3.6,0) -- (3.6, 0.9) -- (4.8, 0.9) -- (4.8, 0.3)  -- (4.2, .3) -- (4.2, 0) -- cycle;

\node (lambda prime) at (2.2,.1) {$\lambda\neq\lambda'$};
\draw[->] (lambda prime) -- (1,.4) ;
\draw[->] (lambda prime) -- (3.45,.4) ;

\end{tikzpicture}

\caption{\label{fig:YoungDiagramsGoodBad}
We use $N=17$ and $M=21$. a)  Example of a ``bad'' irrep $\lambda_N \times \lambda_M$: the shape of the diagrams below the first row for $S_N$ and $S_M$ are the same $\lambda$ ; b) Example of a ``good'' irrep: the shape of the diagram below the first line of $S_N$ is strictly included into the one for $S_M$.}
\end{figure}
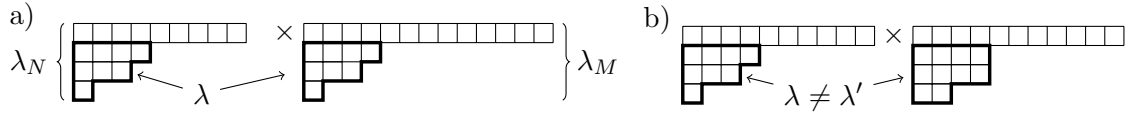

\paragraph{Discussion.}
From Theorem~\ref{thm:adv-representation}, we see that we need to compute the norm of a matrix $\Delta_x^l$ for each irrep $l$ of $G_x=S_{N-1}\times S_M$. We show that these matrices are non-zero only for three different types of irreps of $G_x$. Indeed, for irreps $k$ of $G$ and $l$ of $G_x$, the quantity $\gamma_k\Tr\left[\Pi_y^x\Pi_k\Pi_y^x \Pi^l_{k_1m_1\leftarrow k_2m_2}\right]$ is non-zero only if: \ding{172} $k$ is a bad irrep (otherwise $\gamma_k=0$); \ding{173} $k$ and $l$ restrict to a common irrep of $G_{xy}=S_{N-1}\times S_{M-1}$ (otherwise the product of the projectors is zero). The restrictions of an irrep $(\lambda,\lambda')$ of $G$ to $G_{xy}$ are obtained by removing one box from each of the diagrams $\lambda$ and $\lambda'$. Similarly, the restrictions of an irrep $(\lambda,\lambda')$ of $G_x$ to $G_{xy}$ are obtained by removing one box from $\lambda'$. \ding{174} Note that not all irreps of $G_{xy}$ appear in the projector $\Pi_{y}^x$, as it projects on all injective functions such that $f(x)=y$. Therefore, this set is isomorphic to the set of injective functions from $[N-1]$ to $[M-1]$, and we know that the irrep $\rep$ acting on this set is multiplicity-free, and that only irreps $(\lambda,\lambda')$ where $\lambda$ is contained in $\lambda'$ can occur. Altogether, this implies that only three type of irreps of $G_x=S_{N-1}\times S_M$ lead to non-zero matrices (see Fig.~\ref{fig:discussion})
\begin{enumerate*}
 \item $l=(\lambda,\lambda)$: Same diagram for $S_{N-1}$ and $S_{M}$ below the first row. This irrep has multiplicity one since there is only one way to induce to a valid irrep of $S_N\times S_M$, by adding a box in the first row of the left diagram, leading to irrep $k=(\lambda,\lambda)$.
 \item $l=(\lambda,\lambda^+)$: Diagram for $S_{M}$ has one additional box below the first row. This irrep has multiplicity two since there are two ways to induce to a valid irrep of $S_N\times S_M$, by adding a box either in the first row, leading to $k=(\lambda,\lambda^+)$, or at the missing place below the first row, leading to $k=(\lambda^+,\lambda^+)$.
\item $l=(\lambda,\lambda^{++})$: Diagram for $S_{M}$ has two additional boxes below the first row. This irrep has multiplicity three since there are three ways to induce to a valid irrep of $S_N\times S_M$, by adding a box either in the first row, leading to $k=(\lambda,\lambda^+)$, or at to one of the missing places below the first row, leading to $k=(\lambda^+,\lambda^{++})$.
\end{enumerate*}
Let us now consider these three cases separately.

\newcommand{\scale}{.7}

\newcommand{\lambdaYoung}[1][1]{
\begin{tikzpicture}[scale=#1]
\draw[line width=1.5pt] (0,0) -- (0, 0.9) -- (1.2, 0.9) -- (1.2, 0.6)  -- (0.9, .6) -- (.9, .3) -- (.3, .3) -- (.3, 0) -- cycle;
\draw[step=.3] (0,0) grid (.3, .9);
\draw[step=.3] (.3, .3) grid (.9, .9);
\end{tikzpicture}
}

\newcommand{\lambdaMoins}[1][1]{
\begin{tikzpicture}[scale=#1]
\draw[line width=1.5pt] (0,0) -- (0, 0.9) -- (1.2, 0.9) -- (1.2, 0.6)  -- (.6, .6) -- (.6, .3) -- (.3, .3) -- (.3, 0) -- cycle;
\draw[step=.3] (0,.3) grid (.6, .9);
\draw[step=.3] (.6, .6) grid (1.2, .9);
\end{tikzpicture}
}

\newcommand{\lambdaPlus}[1][1]{
\begin{tikzpicture}[scale=#1]
\draw[line width=1.5pt] (0,0) -- (0, 0.9) -- (1.2, 0.9) -- (1.2, 0.6)  -- (0.9, .6) -- (.9, .3) -- (.6, .3) -- (.6, 0) -- cycle;
\draw[step=.3] (0,0) grid (.3, .9);
\draw[step=.3] (.3, .3) grid (.9, .9);
\end{tikzpicture}
}

\newcommand{\lambdaTilde}[1][1]{
\begin{tikzpicture}[scale=#1]
\draw[line width=1.3pt] (0,0) -- (0, 0.9) -- (1.2, 0.9) -- (1.2, 0.6)  -- (0.6, .6) --  (.6, 0) -- cycle;
\draw[step=.3] (0,0) grid (.6, .9);
\draw[step=.3] (.6, .6) grid (1.2, .9);
\end{tikzpicture}
}

\newcommand{\morebox}[1][1]{
\begin{tikzpicture}[scale=#1] \filldraw[fill=lightgray] (0,0) rectangle (.3, .3); \end{tikzpicture}
}

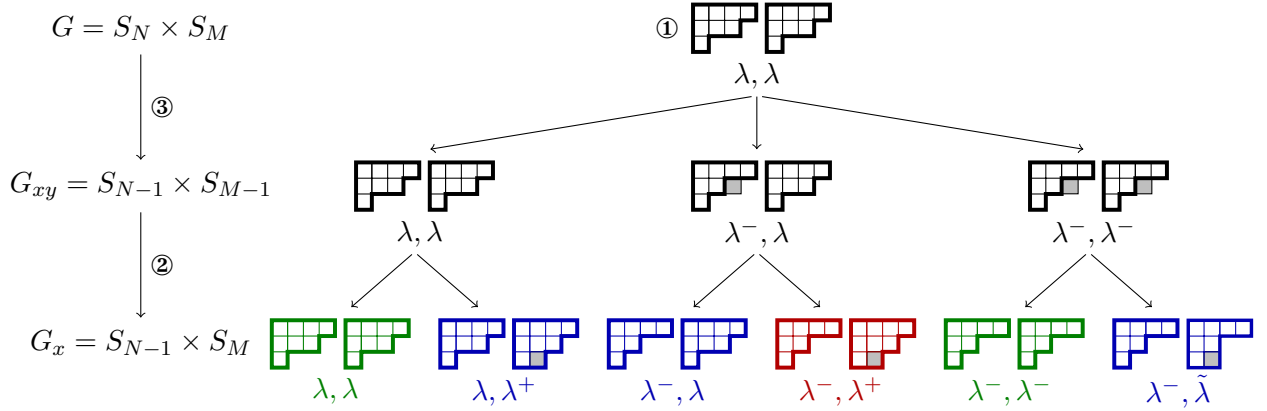
\begin{figure}
\begin{tikzpicture}[scale=\scale]

%%%% SN-1 x SM
\node at (-3,0) {$G_x = S_{N-1} \times S_M$};

\node at (0,0) {\color{green!50!black} \lambdaYoung[\scale]};
\node at (1.4,0) {\color{green!50!black} \lambdaYoung[\scale]};
\node at (.7,-.9) {\color{green!50!black} $\lambda, \lambda$};

\node at (3.2,0) {\color{blue!70!black} \lambdaYoung[\scale]};
\node at (4.45, -.3) {\morebox[\scale]};
\node at (4.6,0) {\color{blue!70!black} \lambdaPlus[\scale]};
\node at (3.9,-.9) {\color{blue!70!black} $\lambda, \lambda^+$};

\node at (6.4, 0) {\color{blue!70!black} \lambdaMoins[\scale]};
\node at (7.8, 0) {\color{blue!70!black} \lambdaYoung[\scale]};
\node at (7.1,-.9) {\color{blue!70!black} $\lambda^-, \lambda$};

\node at (9.6,0)   {\color{red!70!black}  \lambdaMoins[\scale]};
\node at (10.85, -.3) {\morebox[\scale]};
\node at (11,0) {\color{red!70!black} \lambdaPlus[\scale]};
\node at (10.3,-.9) {\color{red!70!black} $\lambda^-, \lambda^+$};

\node at (12.8,0) {\color{green!50!black} \lambdaMoins[\scale]};
\node at (14.2,0) {\color{green!50!black} \lambdaMoins[\scale]};
\node at (13.5,-.9) {\color{green!50!black} $\lambda^-, \lambda^-$};

\node at (16,0) {\color{blue!70!black} \lambdaMoins[\scale]};
\node at (17.25, -.3) {\morebox[\scale]};
\node at (17.4,0) {\color{blue!70!black} \lambdaTilde[\scale]};
\node at (16.7,-.9) {\color{blue!70!black} $\lambda^-, \tilde\lambda$};

%%% SN-1 x SM-1
\node at (-3,3) {$G_{xy} = S_{N-1} \times S_{M-1}$};

\node at (1.6, 3) {\lambdaYoung[\scale]};
\node at (3, 3) {\lambdaYoung[\scale]};
\node at (2.3,2.1) {$\lambda, \lambda$};

\node at (8.16, 3) {\morebox[\scale]};
\node at (8, 3) {\lambdaMoins[\scale]};
\node at (9.4, 3) {\lambdaYoung[\scale]};
\node at (8.7,2.1) {$\lambda^-, \lambda$};

\node at (14.56, 3) {\morebox[\scale]};
\node at (15.96, 3) {\morebox[\scale]};
\node at (14.4, 3) {\lambdaMoins[\scale]};
\node at (15.8, 3) {\lambdaMoins[\scale]};
\node at (15.1,2.1) {$\lambda^-, \lambda^-$};

%%%%SN x SM
\node at (-3,6) {$G = S_{N} \times S_{M}$};

\node at (8, 6) {\lambdaYoung[\scale]};
\node at (9.4, 6) {\lambdaYoung[\scale]};
\node at (8.7,5.1) {$\lambda, \lambda$};
\node at (7, 6) {\ding{172}};

%%%% ARROWS
\draw[->] (8.7, 4.7) -- (8.7, 3.7);
\draw[->] (8.6, 4.7) -- (2.5, 3.7);
\draw[->] (8.8, 4.7) -- (14.9, 3.7);
\draw[->] (2.15, 1.7) -- ( 1, .7);
\draw[->] (2.25, 1.7) -- ( 3.4, .7);
\draw[->] (8.65, 1.7) -- ( 7.5, .7);
\draw[->] (8.75, 1.7) -- ( 9.9, .7);
\draw[->] (14.95, 1.7) -- ( 13.8, .7);
\draw[->] (15.05, 1.7) -- ( 16.2, .7);
\draw[->] (-3, 5.5) -- node[right] {\ding{174}} (-3, 3.5);
\draw[->] (-3, 2.5) -- node[right] {\ding{173}} (-3, .5);
\end{tikzpicture}

\caption{In accordance with our convention, we draw only the part of the diagram below the first row. 
The condition \ding{172} imposes that the two diagrams on top have the same shape. 
From the first to the second row, one should remove one box to each diagram. When the removed box does not belong to the first row, it is show in light gray. The condition \ding{174} imposes that the left diagram is included into the right one.
The condition \ding{173} gives the third row of diagrams. Finally we have 3 ‘‘generic'' types of irreps: case 1 (blue) where the irreps have the same shape; case 2 (green) where the right diagram has one more box; and case 3 (red) where the the right diagram has 2 more boxes. \label{fig:discussion}}
\end{figure}

\paragraph{Case $(\lambda,\lambda)$.} Since this irrep has multiplicity one, we just need to compute a scalar. As an irrep of $S_{N-1}\times S_M$, $(\lambda,\lambda)$ restricts to only one valid irrep of $S_{N-1}\times S_{M-1}$, by removing a box on the first row of the right diagram, therefore this irrep is also labeled $(\lambda,\lambda)$. Inducing from this irrep of $S_{N-1}\times S_{M-1}$ to $S_{N}\times S_{M}$, we obtain three valid irreps, two ``bad'' ones, $(\lambda,\lambda)$ and $(\lambda^+,\lambda^+)$, and a good one, $(\lambda,\lambda^+)$. To differentiate between projectors of irreps of the different groups, we will from now on use superscripts (for example $\Pi_{\lambda,\lambda}^{N,M}$ denotes a projector on the irrep $(\lambda,\lambda)$ of $S_{N}\times S_{M}$). We therefore have from Theorem~\ref{thm:adv-representation}
\begin{align*}
\Delta_x^{\lambda,\lambda}
=&\frac{\gamma_{|\lambda|}}{d_{\lambda,\lambda}^{N-1,M}}\sum_y\Tr\left[\Pi_y^x\Pi_{\lambda,\lambda}^{N,M}\Pi_y^x \Pi_{\lambda,\lambda}^{N-1,M}\right]
+\frac{\gamma_{|\lambda|+1}}{d_{\lambda,\lambda}^{N-1,M}}\sum_y\Tr\left[\Pi_y^x\Pi_{\lambda^+,\lambda^+}^{N,M}\Pi_y^x \Pi_{\lambda,\lambda}^{N-1,M}\right]
-\gamma_{|\lambda|}\\
=&\frac{M\gamma_{|\lambda|}}{d_{\lambda,\lambda}^{N-1,M}d_{\lambda,\lambda}^{N-1,M-1}}\cdot \Tr\left[\Pi_{\lambda,\lambda}^{N-1,M-1}\Pi_{\lambda,\lambda}^{N,M}\right]\cdot
\Tr\left[\Pi_{\lambda,\lambda}^{N-1,M-1} \Pi_{\lambda,\lambda}^{N-1,M}\right]\\
& +\frac{M\gamma_{|\lambda|+1}}{d_{\lambda,\lambda}^{N-1,M}d_{\lambda,\lambda}^{N-1,M-1}}
\cdot\Tr\left[\Pi_{\lambda,\lambda}^{N-1,M-1}\Pi_{\lambda^+,\lambda^+}^{N,M}\right]\cdot
\Tr\left[\Pi_{\lambda,\lambda}^{N-1,M-1} \Pi_{\lambda,\lambda}^{N-1,M}\right]
-\gamma_{|\lambda|},
\end{align*}
where we have used Lemma~\ref{lem:trace-proj} and the fact that all terms in the sum over $y$ are equal by symmetry. We also have
\begin{align*}
 \Tr\left[\Pi_{\lambda,\lambda}^{N-1,M-1}\Pi_{\lambda,\lambda}^{N,M}\right]
&=\Tr\left[\Pi_{\lambda,\lambda}^{N-1,M-1}\Pi_{\lambda,\lambda}^{N-1,M}\right]\\
\Tr\left[\Pi_{\lambda,\lambda}^{N-1,M-1}\Pi_{\lambda^+,\lambda^+}^{N,M}\right]
&=\Tr\left[\Pi_{\lambda,\lambda}^{N-1,M-1}\Pi_{\lambda^+,\lambda}^{N,M-1}\right],
\end{align*}
since the only way for $(\lambda,\lambda)$ as an irrep of $S_{N}\times S_{M}$ to restrict to $(\lambda,\lambda)$ as an irrep of $S_{N-1}\times S_{M-1}$ is to first restrict to $(\lambda,\lambda)$ as an irrep of $S_{N-1}\times S_{M}$, and similarly for $(\lambda^+,\lambda^+)$. Therefore, we only have two traces to compute. For the first one, we consider the maximally mixed state $\rho_{\lambda,\lambda}^{N-1,M-1}$ over the corresponding irrep. By inducing from $S_{M-1}$ to $S_M$ we find that its overlap over the irrep $(\lambda,\lambda)$ of $S_{N-1}\times S_M$ is given by
\begin{align*}
 \Tr\left[\rho_{\lambda,\lambda}^{N-1,M-1}\Pi_{\lambda,\lambda}^{N-1,M}\right]
=\frac{d_{\lambda,\lambda}^{N-1,M}}{Md_{\lambda,\lambda}^{N-1,M-1}}
=\frac{d_{\lambda}^{M}}{Md_{\lambda}^{M-1}}.
\end{align*}
Similarly, we obtain
\begin{align*}
 \Tr\left[\rho_{\lambda,\lambda}^{N-1,M-1}\Pi_{\lambda^+,\lambda}^{N,M-1}\right]
=\frac{d_{\lambda^+,\lambda}^{N,M-1}}{Nd_{\lambda,\lambda}^{N-1,M-1}}
=\frac{d_{\lambda^+}^{N}}{Nd_{\lambda}^{N-1}},
\end{align*}
and finally
\begin{align*}
 \Delta_x^{\lambda,\lambda}
&=\gamma_{|\lambda|}\frac{d_{\lambda}^{M}}{Md_{\lambda}^{M-1}}
+\gamma_{|\lambda|+1}\frac{d_{\lambda^+}^{N}}{Nd_{\lambda}^{N-1}}
-\gamma_{|\lambda|}\\
% &=\gamma_{|\lambda|+1}-\gamma_{|\lambda|}+O(\frac{1}{N})\\
&=\frac{1}{\sqrt{N}}+O(\frac{1}{N}),
\end{align*}
where we have used the hook-length formula for dimensions of irreps and the fact that the number of boxes $|\lambda|$ below the first row is at most $\sqrt{N}$, otherwise $\gamma_{|\lambda|}=0$.

\paragraph{Case $(\lambda,\lambda^+)$.}
This irrep has multiplicity two, so we need to compute a $2\times 2$ matrix.
Let $(\lambda, \lambda^+, 1)$ denote the copy of $(\lambda,\lambda^+)$ irrep 
of $S_{N-1}\times S_M$ which is inside the $(\lambda^+,\lambda^+)$ irrep of $S_N \times S_M$.
Let $(\lambda, \lambda^+, 2)$ denote the copy of $(\lambda,\lambda^+)$ irrep 
of $S_{N-1}\times S_M$ which is inside the $(\lambda,\lambda^+)$ irrep of $S_N \times S_M$.
Let the first row and the first column of $\Delta_x^{\lambda,\lambda^+}$ be indexed
by $(\lambda, \lambda^+, 1)$ and the second row and the second column be indexed by
$(\lambda, \lambda^+, 2)$.

An irrep $(\lambda, \lambda^+)$ of $S_{N-1}\times S_M$ restricts to two valid irreps 
of $S_{N-1}\times S_{M-1}$: $(\lambda, \lambda)$ and $(\lambda, \lambda^+)$.
Those two irreps can be induced to the following bad irreps of $S_N \times S_M$:
$(\lambda, \lambda)$ and any irrep
$(\lambda', \lambda')$ which has one more square below the first row than
$\lambda$. ($\lambda'$ may be equal or different from $\lambda^+$.)

For brevity, we denote $\Delta_x^{\lambda,\lambda^+}$ simply by $\Delta$.
Since $(\lambda, \lambda^+, 1)$ is contained inside a bad irrep of $S_N \times S_M$, we have
\begin{align*} \Delta_{1, 1}
& = \frac{\gamma_{|\lambda|}}{d_{\lambda,\lambda^+}^{N-1,M}}\sum_y\Tr\left[\Pi_y^x\Pi_{\lambda,\lambda}^{N,M}\Pi_y^x \Pi_{\lambda,\lambda^+, 1}^{N-1,M}\right]
+\frac{\gamma_{|\lambda|+1}}{d_{\lambda,\lambda^+}^{N-1,M}}\sum_{\lambda'}
\sum_y\Tr\left[\Pi_y^x\Pi_{\lambda',\lambda'}^{N,M}\Pi_y^x \Pi_{\lambda,\lambda^+, 1}^{N-1,M}\right]
-\gamma_{|\lambda|+1} \\
& = \frac{M \gamma_{|\lambda|}}{d_{\lambda,\lambda^+}^{N-1,M} d_{\lambda,\lambda}^{N-1,M-1}}\Tr\left[\Pi_{\lambda, \lambda}^{N-1, M-1}\Pi_{\lambda,\lambda}^{N,M}\right] \Tr\left[\Pi_{\lambda, \lambda}^{N-1, M-1} \Pi_{\lambda,\lambda^+, 1}^{N-1,M}\right] \\
&+
\frac{M \gamma_{|\lambda|+1}}{d_{\lambda,\lambda^+}^{N-1,M} d_{\lambda,\lambda}^{N-1,M-1}}\left( \sum_{\lambda'}
\Tr\left[\Pi_{\lambda, \lambda}^{N-1, M-1} \Pi_{\lambda',\lambda'}^{N,M} \right] \right) 
\Tr\left[\Pi_{\lambda, \lambda}^{N-1, M-1} \Pi_{\lambda,\lambda^+, 1}^{N-1,M}\right] 
%\\&+
%\frac{M \gamma_{|\lambda|}}{d_{\lambda,\lambda^+}^{N-1,M}d_{\lambda,\lambda^+}^{N-1,M-1}}\Tr\left[\Pi_{\lambda, \lambda^+}^{N-1, M-1}\Pi_{\lambda,\lambda}^{N,M}\right] 
%\Tr\left[\Pi_{\lambda, \lambda^+}^{N-1, M-1} \Pi_{\lambda,\lambda^+, 1}^{N-1,M}\right]
\\& +\frac{M \gamma_{|\lambda|+1}}{d_{\lambda,\lambda^+}^{N-1,M}d_{\lambda,\lambda^+}^{N-1,M-1}}%\left( \sum_{\lambda'}
\Tr\left[\Pi_{\lambda, \lambda^+}^{N-1, M-1}\Pi_{\lambda^+,\lambda^+}^{N,M}\right] %\right) 
\Tr\left[\Pi_{\lambda, \lambda^+}^{N-1, M-1} \Pi_{\lambda,\lambda^+, 1}^{N-1,M}\right]
%=\frac{M\gamma_{|\lambda|}}{d_{\lambda,\lambda}^{N-1,M}d_{\lambda,\lambda}^{N-1,M-1}}\cdot \Tr\left[\Pi_{\lambda,\lambda}^{N-1,M-1}\Pi_{\lambda,\lambda}^{N,M}\right]\cdot
%\Tr\left[\Pi_{\lambda,\lambda}^{N-1,M-1} \Pi_{\lambda,\lambda}^{N-1,M}\right]\\
%& +\frac{M\gamma_{|\lambda|+1}}{d_{\lambda,\lambda}^{N-1,M}d_{\lambda,\lambda}^{N-1,M-1}}
%\cdot\Tr\left[\Pi_{\lambda,\lambda}^{N-1,M-1}\Pi_{\lambda^+,\lambda^+}^{N,M}\right]\cdot
%\Tr\left[\Pi_{\lambda,\lambda}^{N-1,M-1} \Pi_{\lambda,\lambda}^{N-1,M}\right]
-\gamma_{|\lambda|+1}
\end{align*}
We start by evaluating the sum
%\begin{equation}
%\label{eq:sum1}
\[ \sum_{\lambda'}
\Tr\left[\Pi_{\lambda, \lambda}^{N-1, M-1} \Pi_{\lambda',\lambda'}^{N,M} \right]. \]
%\end{equation}
We consider the maximally mixed state $\rho_{\lambda, \lambda}^{N-1, M-1}$ over the 
corresponding irrep of $S_{N-1}\times S_{M-1}$.
By inducing $\lambda$ from $S_{N-1}$ to $S_{N}$, we find that the dimension of the 
induced representation is $N d_{\lambda}^{N-1}$ and the induced representation decomposes
into irrep $\lambda$ of $S_N$, with dimension $d_{\lambda}^N$ and irreps 
$\lambda'$. Therefore, 
\begin{equation}
\label{eq:a1} \sum_{\lambda'}
\Tr\left[\Pi_{\lambda',\lambda'}^{N,M} \rho_{\lambda, \lambda}^{N-1, M-1} \right] =1 - \frac{d_{\lambda}^N}{N d_{\lambda}^{N-1}}
\leq 1-\frac{1}{N}
%  \geq 1 - \frac{1}{N}\left( \frac{N}{N-2|\lambda|}\right)
% = 1 - O\left( \frac{1}{N} \right) ,
\end{equation}
where the inequality follows by comparing the hook-length formulas of $d_{\lambda}^N$
and $d_{\lambda}^{N-1}$. 
Similarly, we have
\begin{equation}
\label{eq:a2}
\Tr\left[\Pi_{\lambda,\lambda}^{N,M} \rho_{\lambda, \lambda}^{N-1, M-1} \right] = O 
\left( \frac{1}{N} \right) .
\end{equation}
We now evaluate a similar quantity for $\rho_{\lambda, \lambda^+}^{N-1, M-1}$. % instead
%of $\rho_{\lambda, \lambda^+}^{N-1, M-1}$.
By inducing $\lambda^+$ from $S_{M-1}$ to $S_{M}$, we find that the dimension of the 
induced representation is $M d_{\lambda^+}^{M-1}$ and the induced representation decomposes
into irrep $\lambda^+$ of $S_M$, with dimension $d_{\lambda}^M$ and irreps 
$\lambda^{++}$ which have one more square below the first row than $\lambda^+$.
Therefore,
\begin{equation}
\label{eq:a3} 
\Tr\left[\Pi_{\lambda^+,\lambda^+}^{N,M} \rho_{\lambda, \lambda^+}^{N-1, M-1} \right] = 
\frac{d_{\lambda}^M}{M d_{\lambda}^{M-1}} = O  \left( \frac{1}{M} \right) .
\end{equation}
By using eqs.~(\ref{eq:a1}), (\ref{eq:a2}) and (\ref{eq:a3}), we have   
\begin{equation}
\label{eq:a4} \Delta_{1, 1}
 = \frac{M \gamma_{|\lambda|+1}}{d_{\lambda,\lambda^+}^{N-1,M}}
\Tr\left[\Pi_{\lambda, \lambda}^{N-1, M-1} \Pi_{\lambda,\lambda^+, 1}^{N-1,M}\right] 
 + O \left( \frac{1}{N} \right) -\gamma_{|\lambda|+1}.
\end{equation}
We have
\[
\Tr\left[\Pi_{\lambda, \lambda}^{N-1, M-1} \Pi_{\lambda,\lambda^+, 1}^{N-1,M}\right] =
\Tr\left[\Pi_{\lambda, \lambda}^{N-1, M-1} \Pi_{\lambda^+,\lambda^+}^{N,M}\right] \]
because the other irreps of $S_{N-1}\times S_M$ contained in the irrep 
$(\lambda^+,\lambda^+)$ of $S_N \times S_M$ have no overlap with the irrep
$(\lambda, \lambda)$ of $S_{N-1}\times S_{M-1}$.
Let $\rho_{\lambda, \lambda}^{N-1, M-1}$ be the completely mixed state
over $(\lambda, \lambda)$. Then,
\[ \Tr\left[\Pi_{\lambda, \lambda}^{N-1, M-1} \Pi_{\lambda^+,\lambda^+}^{N,M}\right] 
= d_{\lambda, \lambda}^{N-1, M-1} 
\Tr\left[ \Pi_{\lambda^+,\lambda^+}^{N,M} \rho_{\lambda, \lambda}^{N-1, M-1}\right] 
 = d_{\lambda, \lambda}^{N-1, M-1} \frac{d_{\lambda^+}^N}{N d_{\lambda}^{N-1}} .\]
Here, the second equality follows by inducing $\lambda$ from $S_{N-1}$ to $S_N$.
We have
\[ d_{\lambda, \lambda}^{N-1, M-1} \frac{d_{\lambda^+}^N}{N d_{\lambda}^{N-1}} =
d_{\lambda}^{N-1}  d_{\lambda}^{M-1} \frac{d_{\lambda^+}^N}{N d_{\lambda}^{N-1}} =
\frac{d_{\lambda}^{M-1} d_{\lambda^+}^N}{N} .\]
By matching up the terms in hook-length formulas, we have
\begin{equation}
\label{eq:hook} 
d_{\lambda}^{M-1} d_{\lambda^+}^N = \left(1+O\left( \frac{1}{N} \right) \right)
\frac{N}{M} d_{\lambda}^{N-1} d_{\lambda^+}^M .
\end{equation}
Therefore, 
\begin{equation}
\label{eq:trace1} 
\Tr\left[\Pi_{\lambda, \lambda}^{N-1, M-1} \Pi_{\lambda^+,\lambda^+}^{N,M}\right] = 
\left(1+O\left( \frac{1}{N} \right) \right) \frac{d_{\lambda, \lambda^+}^{N-1, M}}{M} 
\end{equation}
and
\[ \Delta_{1, 1}
= O \left( \frac{1}{N} \right)
%  = \gamma_{|\lambda|+1} + O \left( \frac{1}{N} \right) -\gamma_{|\lambda|+1} =
% O \left( \frac{1}{\sqrt{N}} \right).
\]  
Similarly to eq.~(\ref{eq:a4}), we have
\begin{equation}
\label{eq:delta22} \Delta_{2, 2}
 = \frac{M \gamma_{|\lambda|+1}}{d_{\lambda,\lambda^+}^{N-1,M}}
\Tr\left[\Pi_{\lambda, \lambda}^{N-1, M-1} \Pi_{\lambda,\lambda^+, 2}^{N-1,M}\right] 
 + O \left( \frac{1}{N} \right) .
\end{equation}
We have 
\[ \Tr\left[\Pi_{\lambda, \lambda}^{N-1, M-1} \Pi_{\lambda,\lambda^+, 2}^{N-1,M}\right] 
= \Tr\left[\Pi_{\lambda, \lambda}^{N-1, M-1} \Pi_{\lambda,\lambda^+}^{N,M}\right] 
= d_{\lambda, \lambda}^{N-1, M-1}
\Tr\left[\Pi_{\lambda,\lambda^+}^{N,M} \rho_{\lambda, \lambda}^{N-1, M-1}\right],\]
because the other irreps of $S_{N-1}\times S_M$ contained in the irrep 
$(\lambda,\lambda^+)$ of $S_N \times S_M$ have no overlap with the irrep
$(\lambda, \lambda)$ of $S_{N-1}\times S_{M-1}$.

By inducing $\lambda$ from $S_{M-1}$ to $S_M$, we get
\begin{equation}
\label{eq:ind1} 
\Tr\left[\Pi_{\lambda,\lambda^+}^{N,M} \rho_{\lambda, \lambda}^{N-1, M-1}\right]
+ \Tr\left[\Pi_{\lambda^+,\lambda^+}^{N,M} \rho_{\lambda, \lambda}^{N-1, M-1}\right] =
\frac{d^M_{\lambda^+}}{M d^{M-1}_{\lambda}} .
\end{equation}
By inducing $\lambda$ from $S_{N-1}$ to $S_N$, we get
\begin{equation}
\label{eq:ind2} 
\Tr\left[\Pi_{\lambda^+,\lambda^+}^{N,M} \rho_{\lambda, \lambda}^{N-1, M-1}\right] =
\frac{d^N_{\lambda^+}}{N d^{N-1}_{\lambda}} .
\end{equation}
By subtracting eq.~(\ref{eq:ind2}) from eq.~(\ref{eq:ind1}), we get
\[ \Tr\left[\Pi_{\lambda, \lambda}^{N-1, M-1} \Pi_{\lambda,\lambda^+}^{N,M}\right]
= \frac{d^M_{\lambda^+} d^{N-1}_{\lambda}}{M} - \frac{d^N_{\lambda^+} d^{M-1}_{\lambda}}{N} .\] 
Because of eq.~(\ref{eq:hook}), 
\begin{equation}
\label{eq:trace2} 
 \Tr\left[\Pi_{\lambda, \lambda}^{N-1, M-1} \Pi_{\lambda,\lambda^+}^{N,M}\right] =
O \left( \frac{d^M_{\lambda^+} d^{N-1}_{\lambda}}{MN} \right) .
\end{equation}
By substituting this into eq.~(\ref{eq:delta22}), we get $\Delta_{2, 2}=O(\frac{1}{N})$.

Last, we have to bound $\Delta_{1, 2}$ and $\Delta_{2, 1}$.
Similarly to eq.~(\ref{eq:a4}), we have
\[ \Delta_{i, j}
 = \frac{M \gamma_{|\lambda|+1}}{d_{\lambda,\lambda^+}^{N-1,M}}
\Tr\left[\Pi_{\lambda, \lambda}^{N-1, M-1} 
\Pi_{\lambda,\lambda^+, i \leftarrow j}^{N-1,M}\right] 
 + O \left( \frac{1}{N} \right) .\]
By using Lemma \ref{lem:trace-proj} and eqs.~(\ref{eq:trace1}) and (\ref{eq:trace2}),
we get
\[ \Delta_{i, j}=O\left( \frac{1}{\sqrt{N}} \right) .\]
We have shown that $\Delta_{i, j}=O(\frac{1}{\sqrt{N}})$ for all $i, j$.
Therefore, $\|\Delta\|=O(\frac{1}{\sqrt{N}})$.

\paragraph{Case $(\lambda,\lambda^{++})$.} This irrep of $S_{N-1}\times S_M$ has multiplicity three, so we need to bound the elements of a $3\times 3$ matrix. Let $(\lambda,\lambda^{++},1)$ denote the copy of the irrep that lies inside the irrep $(\lambda,\lambda^{++})$ of $(S_N\times S_M)$, $(\lambda,\lambda^{++},2)$ be the copy that lies inside the irrep $(\lambda^+,\lambda^{++})$ of $(S_N\times S_M)$, and $(\lambda,\lambda^{++},3)$ be the copy that lies inside the irrep $({\lambda'}^+,\lambda^{++})$ of $(S_N\times S_M)$, where $\lambda^+$ and ${ \lambda'}^+$ correspond to the two different ways a box can be added to $\lambda$. Since these two last copies have exactly the same structure, they can be treated similarly and we really need to compute only 4 different matrix elements (2 diagonal elements and 2 non-diagonal elements). Let us also note that none of these copies are contained in bad irreps of $S_N\times S_M$.

Let us now denote $\Delta_x^{\lambda,\lambda^{++}}$ by $\Delta$, and index the rows and columns of this matrix by the three copies of the irrep. Note that the irrep $(\lambda,\lambda^{++})$ of $S_{N-1}\times S_M$ restricts to three valid irreps of $S_{N-1}\times S_{M-1}$: $(\lambda,\lambda^{++})$, $(\lambda,\lambda^{+})$ and $(\lambda,{\lambda'}^{+})$. Also only these last two irreps induce two bad irreps of $S_N\times S_M$, $(\lambda^{+},\lambda^{+})$ and $({\lambda'}^{+},{\lambda'}^{+})$, respectively. Therefore, we have for the first diagonal element
\begin{align*}
 \Delta_{1,1}&=\frac{\gamma_{|\lambda|+1}}{d_{\lambda,\lambda^{++}}^{N-1,M}}
\sum_y\left\{\Tr\left[\Pi_y^x\Pi_{\lambda^+,\lambda^+}^{N,M}\Pi_y^x \Pi_{\lambda,\lambda^{++}, 1}^{N-1,M}\right]
+\Tr\left[\Pi_y^x\Pi_{{\lambda'}^+,{\lambda'}^+}^{N,M}\Pi_y^x \Pi_{\lambda,\lambda^{++}, 1}^{N-1,M}\right]
\right\}\\
&=\frac{2M\gamma_{|\lambda|+1}d_{\lambda,\lambda^{+}}^{N-1,M-1}}{d_{\lambda,\lambda^{++}}^{N-1,M}}
\Tr\left[\Pi_{\lambda^+,\lambda^+}^{N,M}\rho_{\lambda,\lambda^+}^{N-1,M-1}\right]\cdot
\Tr\left[\Pi_{\lambda,\lambda^{++},1}^{N-1,M}\rho_{\lambda,\lambda^+}^{N-1,M-1}\right].
\end{align*}
Studying as before the overlap of $\rho_{\lambda,\lambda^+}^{N-1,M-1}$ over the irreps of $S_N\times S_M$, we obtain for the two traces
\begin{align}
 \Tr\left[\Pi_{\lambda^+,\lambda^+}^{N,M}\rho_{\lambda,\lambda^+}^{N-1,M-1}\right]
&\leq\frac{d_{\lambda^+}^M}{Md_{\lambda^+}^{M-1}},\\
\Tr\left[\Pi_{\lambda,\lambda^{++},1}^{N-1,M}\rho_{\lambda,\lambda^+}^{N-1,M-1}\right]\label{eq:trace3}
&=\Tr\left[\Pi_{\lambda,\lambda^{++}}^{N,M}\rho_{\lambda,\lambda^+}^{N-1,M-1}\right]
\leq\frac{d_{\lambda^+}^N}{Nd_{\lambda^+}^{N-1}},
\end{align}
and in turn
\begin{align*}
 \Delta_{1,1}&\leq \frac{2M\gamma_{|\lambda|+1}d_{\lambda^{+}}^{N}d_{\lambda^{+}}^{M}}{Nd_{\lambda^{+}}^{N-1}d_{\lambda^{++}}^{M}}=O\left(\frac{1}{MN}\right).
\end{align*}

For the second diagonal element, we find similarly
\begin{align*}
 \Delta_{2,2}&=\frac{\gamma_{|\lambda|+1}}{d_{\lambda,\lambda^{++}}^{N-1,M}}
\sum_y\left\{\Tr\left[\Pi_y^x\Pi_{\lambda^+,\lambda^+}^{N,M}\Pi_y^x \Pi_{\lambda,\lambda^{++}, 2}^{N-1,M}\right]
+\Tr\left[\Pi_y^x\Pi_{{\lambda'}^+,{\lambda'}^+}^{N,M}\Pi_y^x \Pi_{\lambda,\lambda^{++}, 2}^{N-1,M}\right]
\right\}\\
&=\frac{M\gamma_{|\lambda|+1}d_{\lambda,\lambda^{+}}^{N-1,M-1}}{d_{\lambda,\lambda^{++}}^{N-1,M}}
\Tr\left[\Pi_{\lambda^+,\lambda^+}^{N,M}\rho_{\lambda,\lambda^+}^{N-1,M-1}\right]\cdot
\left\{
\Tr\left[\Pi_{\lambda,\lambda^{++},2}^{N-1,M}\rho_{\lambda,\lambda^+}^{N-1,M-1}\right]
+\Tr\left[\Pi_{\lambda,\lambda^{++},2}^{N-1,M}\rho_{\lambda,{\lambda'}^+}^{N-1,M-1}\right]
\right\}\\
&\leq\frac{2\gamma_{|\lambda|+1}d_{\lambda^+}^M}{d_{\lambda^{++}}^M}=O\left(\frac{1}{M}\right),
\end{align*}
where we have used eq.~(\ref{eq:trace3}) and the fact that the other overlaps are at most 1.

Using exactly the same arguments, we find for the non-diagonal elements
\begin{align*}
 |\Delta_{1,2}|&\leq\frac{2\gamma_{|\lambda|+1}d_{\lambda^+}^M}{d_{\lambda^{++}}^M}\sqrt{\frac{d_{\lambda^+}^N}{Nd_{\lambda^{+}}^{N-1}}}
=O\left(\frac{1}{M\sqrt{N}}\right),\\
|\Delta_{2,3}|&\leq\frac{2\gamma_{|\lambda|+1}d_{\lambda^+}^M}{d_{\lambda^{++}}^M}=O\left(\frac{1}{M}\right).
\end{align*}
Since the irreps $(\lambda,\lambda^{++},2)$ and $(\lambda,\lambda^{++},3)$ are of the same type, we also have $\Delta_{3,3}=O(1/M)$ and $\Delta_{1,3}=O(1/(M\sqrt{N}))$. Therefore, all elements of $\Delta$ are at most $O(1/M)$, so that $\norm{\Delta}=O(1/M)$.

Finally, since the matrices corresponding to all irreps have norm at most $O(1/\sqrt{N})$, we have from Theorem~\ref{thm:adv-representation} $\norm{\aGamma_{x}-\aGamma}=O(1/\sqrt{N})$, and in turn
\begin{align*}
 Q_{\varepsilon}(\IE) =
\Omega\left(
(\sqrt{1-\varepsilon}-\sqrt{N/M})^2
\sqrt{N}
\right).
\end{align*}
\end{proof}

\section{Conclusions and outlook}

The hybrid adversary method we introduced in this paper has a strength
that---in a precise, mathematical sense---lies between that of the
known additive and of the multiplicative adversary methods.
In our opinion, our new method combines the advantages of the additive and multiplicative bounds: (i) it is not more complicated to use than the additive method and (ii) it can lead to lower bounds even for cases of algorithms with small success probability, like the multiplicative method. Furthermore, it can prove lower bounds for quantum state generation problems. We have also shown how to leverage the symmetries of a problem to simplify the computation of the adversary bound, using group representation theory. Altogether, this allowed us to prove a new and tight lower bound for the \IE{} problem.

There are several directions for future research that might present
themselves as this point. By clarifying the relation between the different adversary methods, we are one step closer to a proof that the additive bound satisfies a strong direct product theorem like the multiplicative bound. Indeed, our results imply that it is sufficient to prove that whenever the multiplicative adversary method can prove a lower bound in the limit $\lambda\to 1$, there exists some fixed $\lambda>1$ which leads to the same bound. The most important consequence would be for the quantum query complexity of functions, which would itself satisfy a strong direct product theorem for any function, since the additive adversary method is known to be tight in that case~\cite{Rei09,LMRS10}.

As far as \GI{} is
concerned, one natural question to consider is if the methods can
be extended beyond the model considered here. In particular to allow
more powerful oracles that do not have such strong restrictions for
the access to the graphs. One interesting open question is
if a limitation can be shown for any quantum walk based approach to
\GI. The results shown in this paper are a first step in
this direction but significantly new ideas would be necessary.
Finally, there is an open question that touches on the issue of ``junk'': in the paper we showed
lower bounds for the coherent quantum state generation problem. We
conjecture that the extension to the case where some undesired state is generated along with the target state should also be possible, however, we have not been able to establish this result so far.

\section*{Acknowledgments}

The authors thank Ben Reichardt and Robert \v{S}palek for useful comments.
L.M., M.R and J.R. acknowledge support by ARO/NSA under grant W911NF-09-1-0569. 
A.A. and L.M. acknowledge the support of the European Commission IST project ``Quantum Computer Science'' QCS 25596. 
A.A. was also supported by ESF project 1DP/1.1.1.2.0/09/APIA/VIAA/044 and
FP7 Marie Curie International Reintegration Grant PIRG02-GA-2007-224886. 
L.M. also acknowledges the financial support of Agence Nationale de la Recherche under the projects ANR-09-JCJC-0067-01 (CRYQ) and ANR-08-EMER-012 (QRAC).

\bibliography{adversary}

\end{document}